\documentclass[orivec]{llncs}
\usepackage[utf8]{inputenc}

\usepackage[T1]{fontenc}
\usepackage{lmodern}

\usepackage{amsmath}
\usepackage{amssymb}

\usepackage{bm}

\usepackage{mathtools}

\usepackage{xcolor}

\usepackage[shortlabels]{enumitem}

\usepackage{algorithm2e}

\usepackage[
colorlinks=true,
linkcolor=blue,
citecolor=blue
]{hyperref}
\usepackage{color}

\usepackage[save]{silence}
\usepackage[hypcap=true]{subcaption}
\usepackage{tikz}
\usetikzlibrary{automata}
\usetikzlibrary{positioning}
\usetikzlibrary{arrows.meta}
\usetikzlibrary{shapes.multipart}
\usetikzlibrary{shapes.geometric}
\usetikzlibrary{matrix}
\usetikzlibrary{backgrounds}

\usepackage{LTS_tools}

\tikzstyle{imp}=[blue]
\tikzstyle{imp only}=[imp,dashed]

\newlength{\colonheight}
\AtBeginDocument{\settoheight{\colonheight}{:}}
\newcommand{\tikzarrow}[1][]{%
    \tikz[baseline={([yshift=-.5\colonheight]current bounding box.north)}]%
            {\draw[-stealth,#1] (0,0) -- (10pt,0) ;}}

\usepackage[
backend=biber,
citestyle=numeric-comp,
style=lncs
]{biblatex}
\newcommand{\centermath}[1]{\[#1\]}

\usepackage[bibliography=common]{apxproof}

\newlength{\tabwidth}
\newcommand{\tab}[1][\tabwidth]{\hspace*{#1}}
\newenvironment{case_distinction}{\begin{enumerate}[left= 0mm .. \tabwidth,align=left,nosep,itemindent=*]}{\end{enumerate}}

\usepackage{amsthm}
\newtheoremrep{lemma}{Lemma}
\newtheoremrep{theorem}{Theorem}
\newtheoremrep{conjecture}{Conjecture}

\usepackage{cleveref}
\crefname{conjecture}{conjecture}{conjectures}

\def\itemrefInternal#1:#2:#3\relax{\cref{#1:#2}.\ref{item:#2_#3}}
\newcommand*{\itemref}[2]{\itemrefInternal#1:#2\relax}
\def\ItemrefInternal#1:#2:#3\relax{\Cref{#1:#2}.\ref{item:#2_#3}}
\newcommand*{\Itemref}[2]{\ItemrefInternal#1:#2\relax}

\newcommand*{\proofstep}[1]{{\color{blue}(*#1*)}}

\setupcounter{spec}{S}
\setupcounter{imp}{I}

\newcommand{\equalizeCounters}{%
 \ifnum\value{specCounter}>\value{impCounter}
    \setcounter{impCounter}{\value{specCounter}}%
  \else 
    \setcounter{specCounter}{\value{impCounter}}%
\fi
}

\title{Compositionality in Model-Based Testing}

\author{
    Gijs van Cuyck\inst{1}%
    \and
    Lars van Arragon\inst{1}%
    \and
    Jan Tretmans\inst{1,2}%
}
\institute{
    Radboud University,
    Institute iCIS,
    Nijmegen,
    The Netherlands
    \and
    TNO-ESI,
    Eindhoven,
    The Netherlands \\
    \email{\{gijs.vancuyck,lars.vanarragon,jan.tretmans\}@ru.nl}\thanks{
    This work is part of the project
    \textit{TiCToC - Testing in Times of Continuous Change},
    project nr 17936, part of the research program
    \textit{MasCot - Mastering Complexity},
    which is supported by the Dutch Research Council NWO.}
}

\begin{document}

\maketitle

\begin{abstract}
Model-based testing (MBT) promises a scalable solution to testing large systems, if a model is available. Creating these models for large systems, however, has proven to be difficult. Composing larger models from smaller ones could solve this, but our current MBT conformance relation $\uioco$ is not compositional, i.e.\ correctly tested components, when composed into a system, can still lead to a faulty system. To catch these integration problems, we introduce a new relation over component models called \emph{mutual acceptance}. Mutually accepting components are guaranteed to communicate correctly, which makes MBT compositional. In addition to providing compositionality, mutual acceptance has benefits when retesting systems with updated components, and when diagnosing systems consisting of components.
\end{abstract}

\section{Introduction}
\label{sec:introduction}
Modern software systems are becoming increasingly large and complex. Traditional testing scales poorly for systems of these sizes. This causes the development and maintenance of test suites to become costly and time consuming, which slows down the development of new functionality. Model-Based Testing (MBT) is a technique that has been developed to increase the efficiency and effectiveness of testing. With MBT, testers create a model of the system under test from which an MBT tool can then automatically generate and execute test cases. This reduces the problem of creating and maintaining a test suite to creating and maintaining a model of the system under test.

Creating models for complex systems, however, is still difficult and laborious, since often no single person understands the whole system well enough.
A solution is to divide and conquer: the system is decomposed into its components which are modelled and tested separately. This requires that the applied MBT methodology is \emph{compositional}: if each component implementation is correct with respect to its component model, then it can be inferred that the composition of component implementations, i.e.\ the system under test, is correct with respect to the composition of component models, i.e.\ the system model.

In this paper, we investigate compositionality for MBT with labelled transition systems as models, $\uioco$ as the conformance relation, and parallelism modelling component composition \cite{vanderbijl_CompositionalTestingIoco_2004}. 
We define a relation over component specification models, called \emph{mutual acceptance}, which guarantees that components communicate neatly, and that $\uioco$ is preserved under composition. We generalise existing results on compositionality \cite{dealfaro_InterfaceAutomata_2001,frantzen_ModelBasedTestingEnvironmental_2007,vanderbijl_CompositionalTestingIoco_2004} by making less restrictive assumptions and using a composition operator that is associative so that also compositions of more than two components can be easily considered. Moreover, we use the more recent $\uioco$ conformance relation instead of $\ioco$ \cite{tretmans_GoodbyeIoco_2022}. A more detailed comparison with related work can be found in \cref{sec:related_work}.

In addition to compositionality, mutual acceptance also benefits testing evolving systems and software product lines. It enables more effective testing when a component is replaced by an updated version, as will be elaborated in \cref{sec:componentSubstitution}.
Diagnosis is the converse of compositionality: if the whole system has a failure, then diagnosis tries to localise the failure in one of its components; \cref{sec:componentSubstitution} will also discuss the use of mutual acceptance in diagnosis.

\paragraph{Overview}
\Cref{sec:background} contains preliminaries. \Cref{sec:exampleIntro} shows why the current approach to compositional model-based testing is not desirable by means of an example.
\Cref{sec:definitions} formalises what it means for two models to be compatible with each other for use in model-based testing, and defines the mutual acceptance relation $\mutuallyaccepts$. Then \cref{sec:acceptingSystems} goes on to prove that this leads to desirable properties, after which \cref{subsec:acceptingExample} revisits the example. \Cref{sec:componentSubstitution} discusses how these properties also lead to a reduced testing effort when substituting components, and how $\mutuallyaccepts$ can be used in diagnosis. \Cref{sec:related_work} describes some of the large body of related work previously done in the area of compositional model-based testing. Finally, \cref{sec:future_work,sec:conclusion} discuss possible future work and summarise the main results of this paper, respectively. This is an extended version with proofs. The original publication will be available trough Springer. All proofs for lemmas and theorems are given in the appendix.

\section{Preliminaries}
\label{sec:background}
We give the formal definitions for the MBT theory that we consider. We base our work on the theory developed in \cite{tretmans_ModelBasedTesting_2008,vanderbijl_CompositionalTestingIoco_2004}.
The main formalism used is that of labelled transition systems (LTS) (\cref{def:LTS}). An LTS has states and transitions between states that model events. An event can be an input, an output or $\tau$; $\tau$ represents an internal transition which is not observable from the outside and can therefore not be tested. $I_s$, $U_s$, etc., indicate inputs and outputs, respectively, coming from $LTS$ $s$. The shorthand $L_s$ means $I_s \cup U_s$. The name of an $LTS$ is sometimes used as shorthand for its starting state. $\LTS[I,U]{}$ denotes the domain of labelled transition systems with inputs $I$ and outputs $U$, or just $\LTS$ if $I$ and $U$ are known.
For technical reasons we restrict this class to strongly converging and image-finite systems. Strong convergence means that infinite sequences of $\tau$-actions are not allowed to occur. Image-finiteness means that the number of non-deterministically reachable states shall be finite.
In examples, inputs and outputs are given implicitly by prefixing inputs with $?$, and outputs with $!$. The same label can be in the input set of one $LTS$ and in the output set of another.

\begin{definition}
\label{def:LTS}
A \emph{Labelled Transition System} is a 5-tuple $\langle Q,I,U,T,q_0 \rangle$ where:
\begin{itemize}
    \item $Q$ is a non-empty, countable set of states;
    \item $I$ is a countable set of input labels;
    \item $U$ is a countable set of output labels, which is disjoint from $I$;
    \item $T\subseteq Q \times (I\cup U\cup \{\tau\})\times Q$ is a set of triples, the transition relation;
    \item $q_0 \in Q$ is the initial state.
\end{itemize}
\end{definition}

\noindent
Reasoning about labelled transition systems uses the concept of traces. A trace is a sequence of labels that can occur when walking trough an LTS. Common notation used when describing traces is repeated in \cref{def:arrowdefs}.

\begin{definition}
Let $s \in \LTS$; $p_1,\,p_2\in Q_s$; $\ell\in L_s$; $\sigma \in L_s^*$; $\ell_\tau\in L_s\cup\{\tau\}$; $\sigma_\tau\in (L_s\cup\{\tau\})^*$, where $\epsilon$ denotes the empty sequence of labels.
$$\begin{array}[t]{r@{~~~}c@{~~~}l}
    p_1\xrightarrow{\epsilon}p_2 & \defeq & p_1=p_2 \\
    p_1\xrightarrow{\ell_\tau}p_2 & \defeq & (p_1,\ell_\tau,p_2)\in T_s \\
    p_1\xrightarrow{\ell_\tau \cdot \sigma_\tau}p_2 & \defeq & \exists p_3\in Q_s: p_1 \xrightarrow{\ell_\tau} p_3 \land  p_3\xrightarrow{\sigma_\tau}p_2 \\
    p_1 \xrightarrow{\sigma} & \defeq & \exists p_3 \in Q_s : p_1 \xrightarrow{\sigma} p_3 \\
    p_1 \nottrans{\sigma} & \defeq & \nexists p_3 \in Q_s : p_1 \xrightarrow{\sigma} p_3 \\
    p_1\xRightarrow{\epsilon}p_2 & \defeq & \exists \varphi \in \{\tau\}^* : p_1 \xrightarrow{\varphi} p_2\\
    p_1 \xRightarrow{\sigma\cdot\ell}p_2 & \defeq & \exists p_3,p_4 \in Q_s :
      p_1 \xRightarrow{\sigma} p_3 \land
      p_3 \xrightarrow{\ell} p_4 \land
      p_4 \xRightarrow{\epsilon} p_2\\
    p_1 \xRightarrow{\sigma} & \defeq & \exists p_3 \in Q_s :
     p_1 \xRightarrow{\sigma} p_3 \\
    p_1 \Nottrans{\sigma} & \defeq & \nexists p_3 \in Q_s : p_1 \xRightarrow{\sigma} p_3 \\
    \end{array}$$
\label{def:arrowdefs}
\end{definition}

\noindent
While specifications are often given as an $LTS$, $IOTS$ are used to represent implementations. In MBT, we commonly assume that we can always give any input to an implementation. $\IOTS$ denotes the domain of all input-enabled transition systems, and  $\IOTS[I,U]$ denotes the domain of all input enabled transition systems with input set $I$ and output set $U$.
 
\begin{definition}
\label{def:iots}
$i\in \LTS$ is an \emph{Input-Enabled Transition System} ($IOTS$) if in every state for every input, its transition relation either contains that input, or reaches with just internal transitions another state that does so:
\centermath{\forall q \in Q_i,~ \ell \in I_i:~ q \xRightarrow{\ell}}
\end{definition}

\noindent
Multiple labelled transition systems can be composed to form larger models. For component specifications this is often done using parallel composition (\cref{def:parcomp}). The result of parallel composition represents a system where all the components are being executed at the same time independently of each other. Synchronisation occurs on shared labels. An overview of the label sets of a parallel composition is shown in \cref{fig:parcomp_schema}. Note that we do not require the input sets of the components to be disjoint, which will be explained below. Parallel composition assumes synchronous communication between components. Systems with asynchronous communication can still be modelled, but this requires giving explicit specification for the communication medium.

\begin{definition}
\label{def:composable}
$s, e\in \LTS$ are \emph{composable} iff their respective output sets $U_s$ and $U_e$ are disjoint:~~ $U_s \cap U_e = \emptyset$
\end{definition}

\begin{definition}
\emph{Parallel composition} $\parcomp$ on two composable labelled transition systems $s$ and $e$ is defined as:~
$s\parcomp e ~\defeq~ \langle Q,I,U,T,q_0 \rangle$,~ where 
\begin{itemize}
    \item $Q ~=~ \{\;p\parcomp q \:\setbar\: p\in Q_s, q\in Q_e\;\}$
    \item $I\hspace{3pt} ~=~ (I_s \setminus U_e) \cup (I_e \setminus U_s)$
    \item $U ~=~ U_s\cup U_e$
    \item $q_0 ~=~ {q_0}_s\parcomp {q_0}_e$
    \item $T$ is the minimal set satisfying the following inference rules \\
    (where $p,p_1,p_2 \in Q_s, q,q_1,q_2 \in Q_e$):
    $$\begin{array}[t]{l@{~}r@{~~~~}c@{~~~~}l}
        p_1 \xrightarrow{\ell} p_2 &
        \ell \in (L_s \cup \{\tau\})\setminus L_e &
        \vdash &
        p_1 \parcomp q_{\hphantom{1}} \xrightarrow{\ell} p_2 \parcomp q \\
        q_1 \xrightarrow{\ell} q_2 &
        \ell \in (L_e \cup \{\tau\})\setminus L_s &
        \vdash &
        p_{\hphantom{1}} \parcomp q_1 \xrightarrow{\ell} p_{\hphantom{1}} \parcomp q_2 \\
        p_1 \xrightarrow{\ell} p_2,~ q_1 \xrightarrow{\ell} q_2 &
        \ell \in L_s \cap L_e &
        \vdash &
        p_1 \parcomp q_1 \xrightarrow{\ell} p_2 \parcomp q_2 \\
    \end{array}$$
\end{itemize}
\label{def:parcomp}
\end{definition}

\begin{figure}[t!]
    \centering
    \includegraphics[scale=1.7]{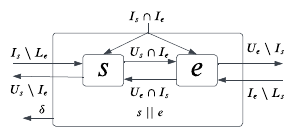}
    \caption{Parallel composition of system $s$ and its environment $e$.}
    \label{fig:parcomp_schema}
\end{figure}

\begin{toappendix}
    \begin{definition}
     $s,s'\in \LTS$ are \textbf{isomorphic} iff:\\ 
    $U_s = U_{s'} \land I_s = I_{s'}$, and there exists a bijective function $f:: Q_s \rightarrow Q_{s'}$ such that $f(s_0) = {s'}_0$ and $\forall p,p'\in Q_s, \ell \in L_s^\delta\cup\{\tau\}: p \xrightarrow{\ell} p' \iff f(p) \xrightarrow{\ell} f(p')$
    \label{def:isomorphism}
    \end{definition}

    \begin{lemma} This lemma is a part of the proof for \cref{lemma:parcomp_assoc}\\
    $\forall a \in  ((U_s\cap L_e) \cup (L_s \cap U_e))\setminus L_t : p \nottrans{a} \lor\; q \nottrans{a}\;\land$\\
    $\forall a \in ((U_s \cup U_e) \cap L_t) \cup ((L_s\cup L_e) \cap U_t) : r \nottrans{a} \lor\; ($\\
    \tab$a \in ((L_s\cap U_t) \cup (U_s \cap L_t))\setminus L_e \implies p \nottrans{a}\;\land$\\
    \tab$a \in ((U_e\cap L_t)\cup (L_e \cap U_t))\setminus L_s \implies q \nottrans{a}\;\land$\\
    \tab$a \in (U_s\cap L_e \cap L_t) \cup (L_s\cap U_e \cap L_t) \cap (L_s \cap L_e \cap U_t) \implies (p \nottrans{a}\lor\; q \nottrans{a}))\;\implies$\\
    
    $\forall b \in (U_s \cap (L_e \cup L_t))\cup(L_s \cap (U_e \cup U_t)) : p \nottrans{b} \lor\; ($\\
    \tab$b \in ((U_s\cap L_e) \cup (L_s \cap U_e))\setminus L_t \implies q \nottrans{b}\;\land$\\
    \tab$b \in ((U_s\cap L_t) \cup (L_s \cap U_t))\setminus L_e \implies r \nottrans{b}\;\land$\\
    \tab$b \in (U_s\cap L_e \cap L_t) \cup (L_s \cap U_e \cap L_t) \cup (L_s \cap L_u \cap U_t) \implies (q \nottrans{b}\lor\; r \nottrans{b}))$\\
    \label{lemma:parcomp_assoc_sublemma_1}
    \end{lemma}
    \begin{proof}
\ \\
$\forall a \in  ((U_s\cap L_e) \cup (L_s \cap U_e))\setminus L_t: p \nottrans{a} \lor\; q \nottrans{a}\;\land$\\
$\forall a \in ((U_s \cup U_e) \cap L_t) \cup ((L_s\cup L_e) \cap U_t) : r \nottrans{a} \lor\; ($\\
\tab$a \in ((L_s\cap U_t) \cup (U_s \cap L_t))\setminus L_e \implies p \nottrans{a}\;\land$\\
\tab$a \in ((U_e\cap L_t)\cup (L_e \cap U_t))\setminus L_s \implies q \nottrans{a}\;\land$\\
\tab$a \in (U_s\cap L_e \cap L_t) \cup (L_s\cap U_e \cap L_t) \cap (L_s \cap L_e \cap U_t) \implies (p \nottrans{a}\lor\; q \nottrans{a}))$\\
\proofstep{Case distinction on $b$}

\begin{case_distinction}
    \item[$b \in  ((U_s\cap L_e) \cup (L_s \cap U_e))\setminus L_t$:]\ \\
    This gives $p \nottrans{b} \lor\; q \nottrans{b}$. $p \nottrans{b}$ Immediately gives the goal, so assume $q \nottrans{b}$.\\
    We have $b \notin L_t$, which makes two of the three preconditions of the goal conjunction false. The third one follows from $q \nottrans{b}$. So all in all:\\
    $True \implies q \nottrans{b}\;\land$\\
    $False \implies r \nottrans{b}\;\land$\\
    $False \implies (q \nottrans{b}\lor r\nottrans{b})$\\
    \proofstep{$\implies$}\\
    $b \in ((U_s\cap L_e) \cup (L_s \cap U_e))\setminus L_t \implies q \nottrans{b}\;\land$\\
    $b \in ((U_s\cap L_t) \cup (L_s \cap U_t))\setminus L_e \implies r \nottrans{b}\;\land$\\
    $b \in (U_s\cap L_e \cap L_t) \cup (L_s \cap U_e \cap L_t) \cup (L_s \cap L_u \cap U_t)  \implies (q \nottrans{b}\lor r\nottrans{b})$\\
    \proofstep{$\implies$}\\
    $p\nottrans{b}\; \lor($\\
    \tab$b \in ((U_s\cap L_e) \cup (L_s \cap U_e))\setminus L_t \implies q \nottrans{b}\;\land$\\
    \tab$b \in ((U_s\cap L_t) \cup (L_s \cap U_t))\setminus L_e \implies r \nottrans{b}\;\land$\\
    \tab$b \in (U_s\cap L_e \cap L_t) \cup (L_s \cap U_e \cap L_t) \cup (L_s \cap L_u \cap U_t)  \implies (q \nottrans{b}\lor r\nottrans{b}))$\\

    \item[$b \in  ((U_s\cap L_t) \cup (L_s \cap U_t))$:]\ \\
    $ r \nottrans{b} \lor\; ($\\
    \tab$b \in ((L_s\cap U_t) \cup (U_s \cap L_t))\setminus L_e \implies p \nottrans{b}\;\land$\\
    \tab$b \in ((U_e\cap L_t)\cup (L_e \cap U_t))\setminus L_s \implies q \nottrans{b}\;\land$\\
    \tab$b \in (U_s\cap L_e \cap L_t) \cup (L_s\cap U_e \cap L_t) \cap (L_s \cap L_e \cap U_t) \implies (p \nottrans{b}\lor\; q \nottrans{b}))$\\
    \proofstep{Case distinction on $r \nottrans{b}$}
    \begin{case_distinction}
        \item[$r \nottrans{b}$:]\ \\
        Similarly to the previous case, one of the assumptions in the goal conjunction is now false due to $b\in L_t$. This gives:\\
        $r \nottrans{b}$\\
        \proofstep{$\iff$}\\
        $False \implies q \nottrans{b}\;\land$\\
        $True \implies r \nottrans{b}\;\land$\\
        $True \implies  r\nottrans{b}$\\
        \proofstep{$\iff$}\\
        $False \implies q \nottrans{b}\;\land$\\
        $True \implies r \nottrans{b}\;\land$\\
        $True \implies (q \nottrans{b}\lor r\nottrans{b})$\\
        \proofstep{$\iff$}\\
        $b \in ((U_s\cap L_e) \cup (L_s \cap U_e))\setminus L_t \implies q \nottrans{b}\;\land$\\
        $b \in ((U_s\cap L_t) \cup (L_s \cap U_t))\setminus L_e \implies r \nottrans{b}\;\land$\\
        $b \in (U_s\cap L_e \cap L_t) \cup (L_s \cap U_e \cap L_t) \cup (L_s \cap L_u \cap U_t)  \implies (q \nottrans{b}\lor r\nottrans{b})$\\
        \proofstep{$\implies$}\\
        $p\nottrans{b}\; \lor($\\
        \tab$b \in ((U_s\cap L_e) \cup (L_s \cap U_e))\setminus L_t \implies q \nottrans{b}\;\land$\\
        \tab$b \in ((U_s\cap L_t) \cup (L_s \cap U_t))\setminus L_e \implies r \nottrans{b}\;\land$\\
        \tab$b \in (U_s\cap L_e \cap L_t) \cup (L_s \cap U_e \cap L_t) \cup (L_s \cap L_u \cap U_t)  \implies (q \nottrans{b}\lor r\nottrans{b}))$\\

        \item[$r \trans{b}$:]\ \\
        $b \in ((L_s\cap U_t) \cup (U_s \cap L_t))\setminus L_e \implies p \nottrans{b}\;\land$\\
        $b \in (U_s\cap L_e \cap L_t) \cup (L_s\cap U_e \cap L_t) \cap (L_s \cap L_e \cap U_t) \implies (p \nottrans{b}\lor\; q \nottrans{b})$\\
        \proofstep{One last case distinction on $b\in L_e$}\\
        If $b\notin L_e$, then we have $p\nottrans{b}$ Which gives the proof goal directly.
        If $b\in L_e$ we have:\\
        $p \nottrans{b}\lor\; q \nottrans{b}$\\
        \proofstep{$\iff$}\\
        $p \nottrans{b}\lor\;($\\
        \tab$False \implies q\nottrans{b}\;\land$\\
        \tab$False \implies r\nottrans{b}\;\land$\\
        \tab$True \implies q \nottrans{b})$\\
        \proofstep{$\implies$}\\
        $p \nottrans{b}\lor\;($\\
        \tab$False \implies q\nottrans{b}\;\land$\\
        \tab$False \implies r\nottrans{b}\;\land$\\
        \tab$True \implies q \nottrans{b}\lor\; r \nottrans{b})$\\
        \proofstep{$\iff$}\\
        $p \nottrans{b}\lor\;($\\
        \tab$b \in ((U_s\cap L_e) \cup (L_s \cap U_e))\setminus L_t \implies q\nottrans{b}\;\land$\\
        \tab$b \in ((U_s\cap L_t) \cup (L_s \cap U_t))\setminus L_e \implies r\nottrans{b}\;\land$\\
        \tab$b \in (U_s\cap L_e \cap L_t) \cup (L_s \cap U_e \cap L_t) \cup (L_s \cap L_u \cap U_t) \implies q \nottrans{b}\lor\; r \nottrans{b})$\\

    \end{case_distinction}

\end{case_distinction}

    \end{proof}

\end{toappendix}

\begin{lemmarep}
Parallel composition is \emph{commutative} and \emph{associative} (up to isomorphism $\equiv$), i.e.\ for $s,e,t \in \LTS$, we have:
    $$\begin{array}[t]{l@{~~~~}r@{~~}c@{~~}l}
    commutativity:  & s \parcomp e              & \equiv & e \parcomp s \\
    associativity: & (s \parcomp e) \parcomp t & \equiv & s \parcomp (e \parcomp t) \\
    \end{array}$$

\label{lemma:parcomp_assoc}
\label{lemma:parcomp_comm}
\end{lemmarep}
    
\begin{proof}
~\begin{case_distinction}
    \item[$commutativity:$]\ \\
    $s\parcomp e$ is defined $\iff$ $e \parcomp s$ is defined. Follows directly from \cref{def:composable}, which is symmetric.
    We prove $s\parcomp e $ and $e \parcomp s$ are isomorphic (\cref{def:isomorphism}).
    $I_{s \parcomp e} \equiv I_{e \parcomp s}$ and $U_{s \parcomp e} \equiv U_{e \parcomp s}$ follow directly from \cref{def:parcomp}. Take $f(p\parcomp q) = q \parcomp p$ for $p \in Q_s, q \in Q_e$. This gives $f(s\parcomp e) = e\parcomp s$ by definition.

    To prove: $\forall p,p' \in Q_s, q,q' \in Q_e, \ell \in Q_{s\parcomp e}^\delta\cup\{\tau\}: p\parcomp q \xrightarrow{\ell} p'\parcomp q' \iff q\parcomp p \xrightarrow{\ell} p' \parcomp q'$.\\
    The if and only if cases are symmetric, and both follow directly from \cref{def:parcomp}  as\\
    $p\parcomp q \xrightarrow{\ell} p' \parcomp q'$ and $q \parcomp p \xrightarrow{\ell} q'\parcomp p'$ have the same preconditions.

    \item[$associativity:$]\ \\
    We first prove $(s \parcomp e)\parcomp t$ is defined $\iff$ $s \parcomp (e\parcomp t)$ is defined:\\
    $(s \parcomp e)\parcomp t$ is defined.\\
    \proofstep{\cref{def:parcomp,def:composable}: $\parcomp$ and $\composable$}\\
    $U_s \cap U_e = U_{s \parcomp e} \cap U_t = \emptyset$\\
    \proofstep{\cref{def:parcomp}: $U_{s \parcomp e}$}\\
    $U_s \cap U_e = (U_s \cup U_e) \cap U_t = \emptyset$\\
    \proofstep{Set rewriting}\\
    $U_s \cap U_e = U_s \cap U_t = U_e \cap U_t = \emptyset$\\
    \proofstep{Set rewriting}\\
    $U_s \cap (U_e \cup U_t) = U_e \cap U_t = \emptyset$\\
    \proofstep{\cref{def:parcomp}: $U_{e \parcomp t}$}\\
    $U_s \cap U_{e \parcomp t} = U_e \cap U_t = \emptyset$\\
    \proofstep{\cref{def:parcomp,def:composable}: $\parcomp$ and $\composable$}\\
    $s \parcomp (e\parcomp t)$ is defined.\\
    \\
    We have $U_{(s\parcomp e)\parcomp t} = U_{s\parcomp e} \cup U_t = U_s \cup U_e \cup U_t = U_s \cup U_{e\parcomp t} = U_{s \parcomp (e \parcomp t)}$ and\\
    \begin{align*}
     I_{(s\parcomp e)\parcomp t} &=\\
     (I_{s\parcomp e} \setminus U_t) \cup (I_t \setminus U_{s\parcomp e}) &= \\
    ((I_s \setminus U_e \cup I_e \setminus U_s) \setminus U_t) \cup (I_t \setminus (U_s \cup U_e)) &= \\
    I_s \setminus (U_e \cup U_t) \cup I_e \setminus (U_s\cup U_t) \cup I_t \setminus (U_s \cup U_e) &=\\
     I_s \setminus (U_e \cup U_t) \cup ((I_e \setminus U_t \cup I_t \setminus U_e)\setminus U_t)  &=\\
    (I_s \setminus U_{e\parcomp t}) \cup (I_{e\parcomp t} \setminus U_t) &= \\
    I_{s \parcomp (e \parcomp t)} &=
    \end{align*}

    Take $f((p\parcomp q)\parcomp r) = p \parcomp (q\parcomp r)$ for $p \in Q_s, q \in Q_e, r \in Q_t$. This gives $f((s\parcomp e) \parcomp t) = s\parcomp (e \parcomp t)$ by definition. To prove: $\forall p,p' \in Q_s, q,q' \in Q_e, r,r' \in Q_r, \ell \in L_{(s\parcomp e)\parcomp t}^\delta\cup\{\tau\}:$\\
    $(p\parcomp q)\parcomp r \xrightarrow{\ell} (p'\parcomp q') \parcomp r' \iff p \parcomp (q\parcomp r) \xrightarrow{\ell} p' \parcomp (q' \parcomp r')$.\\
    $\implies$ and $\impliedby$ cases are symmetrical, so only the $\implies$ case is given here.
    We do a case distinction on $\ell$. Most cases are identical and therefore omited.
    \begin{case_distinction}
        \item[$\ell \in (L_s \cap  L_e) \setminus L_t$:]\ \\
        \proofstep{\cref{def:parcomp}: $\parcomp$}\\
        $p \parcomp q \xrightarrow{\ell} p' \parcomp q' \land r = r'$\\
        \proofstep{\cref{def:parcomp}: $\parcomp$}\\
        $p \trans{\ell} p' \land q \trans{\ell} q' \land r = r'$\\
        \proofstep{\cref{def:parcomp}: $\parcomp$}\\
        $p \trans{\ell} p' \land q \parcomp r \trans{\ell} q' \parcomp r'$\\
        \proofstep{\cref{def:parcomp}: $\parcomp$}\\
        $p \parcomp (q\parcomp r) \xrightarrow{\ell} p' \parcomp (q' \parcomp r')$
        
        \item[$\ell \in (L_s \cap L_t) \setminus L_e$:]\ \\
        \proofstep{\cref{def:parcomp}: $\parcomp$}\\
        $p \parcomp q \trans{\ell} p' \parcomp q' \land r \trans{\ell} r'$\\
        \proofstep{\cref{def:parcomp}: $\parcomp$}\\
        $p \trans{\ell} p' \land q = q' \land r \trans{\ell} r'$\\
        \proofstep{\cref{def:parcomp}: $\parcomp$}\\
        $p \trans{\ell} p' \land q \parcomp r \trans{\ell} q' \parcomp r'$\\
        \proofstep{\cref{def:parcomp}: $\parcomp$}\\
        $p \parcomp (q\parcomp r) \xrightarrow{\ell} p' \parcomp (q' \parcomp r')$
        \item[$\ell \in (L_e \cap L_t) \setminus L_s$,]
        \item[$\ell \in L_s \setminus (L_e\cup L_t)$,]
        \item[$\ell \in L_e \setminus (L_s\cup L_t)$,]
        \item[$\ell \in L_t \setminus (L_s\cup L_e)$,]
        \item[$\ell \in L_s\cap L_e\cap L_t$ :] All similar to previous cases.
        \item[$\ell = \tau$:]\ \\
        \proofstep{\cref{def:parcomp}: $\parcomp$}\\
        $(p\parcomp q \trans{\tau} p'\parcomp q' \land r = r') \lor (r \trans{\tau} r'\land p\parcomp q = p' \parcomp q')$\\
        \proofstep{\cref{def:parcomp}: $\parcomp$}\\
        $(p \trans{\tau} p' \land q = q' \land r = r')\; \lor$\\
        $(q \trans{\tau} q' \land p=p' \land r = r')\;\lor$\\
        $(r \trans{\tau} r'\land p= p' \land q =  q')$\\
        \proofstep{\cref{def:parcomp}: $\parcomp$}\\
        $(p \trans{\tau} p' \land q \parcomp r = q' \parcomp r')\; \lor$\\
        $(q \parcomp r \trans{\tau} q' \parcomp r' \land p=p')\;\lor$\\
        $(q \parcomp r \trans{\tau} q'\parcomp r'\land p= p')$\\
        \proofstep{\cref{def:parcomp}: $\parcomp$}\\
        $(p \parcomp (q\parcomp r) \trans{\tau} p' \parcomp (q'\parcomp r'))\; \lor$\\
        $(p \parcomp (q\parcomp r) \trans{\tau} p' \parcomp (q'\parcomp r'))\;\lor$\\
        $(p \parcomp (q\parcomp r) \trans{\tau} p' \parcomp (q'\parcomp r'))$\\
        \proofstep{$\iff$}\\
        $p \parcomp (q\parcomp r) \xrightarrow{\ell} p' \parcomp (q' \parcomp r')$
        
        \item[$\ell = \delta$:]\ \\
        $(p\parcomp q)\parcomp r \xrightarrow{\delta} \land\; p = p' \land q = q' \land r = r'$\\
        \proofstep{remember $p = p' \land q= q' \land r = r'$ as subproof, dropped here for brevity}\\
        $(p\parcomp q)\parcomp r \xrightarrow{\delta}$\\
        \proofstep{\cref{def:parcomp}:$\parcomp$}\\
        $\forall a \in (U_{s\parcomp e} \setminus L_t) \cup \{\tau\}: p\parcomp q \nottrans{a} \land$\\
        $\forall a \in (U_t \setminus L_{s\parcomp e}) \cup \{\tau\}: r \nottrans{a} \land$\\
        $\forall a \in (U_{s\parcomp e} \cap L_t) \cup (L_{s\parcomp e} \cap U_t): p \parcomp q \nottrans{a} \lor\; r \nottrans{a}$\\
        \proofstep{\cref{def:parcomp}: $U_{s\parcomp e}$ and $L_{s\parcomp e}$}\\
        $\forall a \in (U_s \setminus L_t) \cup (U_e \setminus L_t) \cup \{\tau\}: p\parcomp q \nottrans{a} \land$\\
        $\forall a \in (U_t \setminus (L_s \cup L_e)) \cup \{\tau\}: r \nottrans{a} \land$\\
        $\forall a \in ((U_s \cup U_e) \cap L_t) \cup ((L_s\cup L_e) \cap U_t): p \parcomp q \nottrans{a} \lor\; r \nottrans{a}$\\
        \proofstep{\cref{def:parcomp}: $\parcomp$}\\
        $\forall a \in  (U_s \setminus (L_e \cup L_t)) \cup \{\tau\}: p \nottrans{a}\;\land$\\
        $\forall a \in  (U_e \setminus (L_s \cup L_t)) \cup \{\tau\}: q \nottrans{a}\;\land$\\
        $\forall a \in  ((U_s \cap L_e) \cup (U_e \cap L_s)) \setminus L_t: p \nottrans{a} \lor\; q \nottrans{a}\;\land$\\
        $\forall a \in (U_t \setminus (L_s \cup L_e)) \cup \{\tau\}: r \nottrans{a} \land$\\
        $\forall a \in ((U_s \cup U_e) \cap L_t) \cup ((L_s\cup L_e) \cap U_t): p \parcomp q \nottrans{a} \lor\; r \nottrans{a}$\\
        \proofstep{\cref{def:parcomp}: $\parcomp$}\\
        $\forall a \in  (U_s \setminus (L_e \cup L_t)) \cup \{\tau\}: p \nottrans{a}\;\land$\\
        $\forall a \in  (U_e \setminus (L_s \cup L_t)) \cup \{\tau\}: q \nottrans{a}\;\land$\\
        $\forall a \in  ((U_s \cap L_e) \cup (U_e \cap L_s)) \setminus L_t: p \nottrans{a} \lor\; q \nottrans{a}\;\land$\\
        $\forall a \in (U_t \setminus (L_s \cup L_e)) \cup \{\tau\}: r \nottrans{a} \land$\\
        $\forall a \in ((U_s \cup U_e) \cap L_t) \cup ((L_s\cup L_e) \cap U_t) : r \nottrans{a} \lor\; ($\\
        \tab$a \in ((L_t\cap U_s) \cup (U_t \cap L_s))\setminus L_e \implies p \nottrans{a}\;\land$\\
        \tab$a \in ((U_t\cap L_e)\cup (L_t \cap U_e))\setminus L_s \implies q \nottrans{a}\;\land$\\
        \tab$a \in (U_s\cap L_e \cap L_t) \cup (L_s\cap U_e \cap L_t) \cap (L_s \cap L_e \cap U_t) \implies (p \nottrans{a}\lor\; q \nottrans{a})$\\
        \proofstep{\cref{lemma:parcomp_assoc_sublemma_1}: subproof with many case distinctions in separate lemma.}\\
        $\forall a \in  (U_s \setminus (L_e \cup L_t)) \cup \{\tau\}: p \nottrans{a}\;\land$\\
        $\forall a \in  (U_e \setminus (L_s \cup L_t)) \cup \{\tau\}: q \nottrans{a}\;\land$\\
        $\forall a \in (U_t \setminus (L_s \cup L_e)) \cup \{\tau\}: r \nottrans{a} \land$\\
        $\forall a \in ((U_s \cup U_e) \cap L_t) \cup ((L_s\cup L_e) \cap U_t) : r \nottrans{a} \lor\; ($\\
        \tab$a \in ((L_t\cap U_s) \cup (U_t \cap L_s))\setminus L_e \implies p \nottrans{a}\;\land$\\
        \tab$a \in ((U_t\cap L_e)\cup (L_t \cap U_e))\setminus L_s \implies q \nottrans{a}\;\land$\\
        \tab$a \in (U_s\cap L_e \cap L_t) \cup (L_s\cap U_e \cap L_t) \cap (L_s \cap L_e \cap U_t) \implies (p \nottrans{a}\lor\; q \nottrans{a})$\\
        $\forall a \in (U_s \cap (L_e \cup L_t))\cup(L_s \cap (U_e \cup U_t)) : p \nottrans{a} \lor\; ($\\
        \tab$a \in ((U_s\cap L_e) \cup (L_s \cap U_e))\setminus L_t \implies q \nottrans{a}\;\land$\\
        \tab$a \in ((U_s\cap L_t) \cup (L_s \cap U_t))\setminus L_e \implies r \nottrans{a}\;\land$\\
        \tab$a\in (U_s\cap L_e \cap L_t) \cup (L_s \cap U_e \cap L_t) \cup (L_s \cap L_u \cap U_t) \implies (q \nottrans{b}\lor\; r \nottrans{b}))$\\
        \proofstep{Shrink domain of fourth assumption}\\
        $\forall a \in  (U_s \setminus (L_e \cup L_t)) \cup \{\tau\}: p \nottrans{a}\;\land$\\
        $\forall a \in  (U_e \setminus (L_s \cup L_t)) \cup \{\tau\}: q \nottrans{a}\;\land$\\
        $\forall a \in (U_t \setminus (L_s \cup L_e)) \cup \{\tau\}: r \nottrans{a} \land$\\
        $\forall a \in ((U_e \cap L_t) \cup (L_e \cap U_t))\setminus L_s : r \nottrans{a} \lor\; ($\\
        \tab$False \implies p \nottrans{a}\;\land$\\
        \tab$True\implies q \nottrans{a}\;\land$\\
        \tab$False \implies (p \nottrans{a}\lor\; q \nottrans{a})$\\
        $\forall a \in (U_s \cap (L_e \cup L_t))\cup(L_s \cap (U_e \cup U_t)) : p \nottrans{a} \lor\; ($\\
        \tab$a \in ((U_s\cap L_e) \cup (L_s \cap U_e))\setminus L_t \implies q \nottrans{a}\;\land$\\
        \tab$a \in ((U_s\cap L_t) \cup (L_s \cap U_t))\setminus L_e \implies r \nottrans{a}\;\land$\\
        \tab$a\in (U_s\cap L_e \cap L_t) \cup (L_s \cap U_e \cap L_t) \cup (L_s \cap L_u \cap U_t) \implies (q \nottrans{b}\lor\; r \nottrans{b}))$\\
        \proofstep{$\iff$}\\
        $\forall a \in (U_s \setminus (L_e \cup L_t)) \cup \{\tau\}: p \nottrans{a} \land$\\
        $\forall a \in  (U_e \setminus (L_s \cup L_t)) \cup \{\tau\}: q \nottrans{a}\;\land$\\
        $\forall a \in  (U_t \setminus (L_s \cup L_e)) \cup \{\tau\}: r \nottrans{a}\;\land$\\
        $\forall a \in  ((U_e \cap L_t) \cup (U_t \cap L_e)) \setminus L_s: q \nottrans{a} \lor\; r \nottrans{a}\;\land$\\
        $\forall a \in (U_s \cap (L_e \cup L_t))\cup(L_s \cap (U_e \cup U_t)) : p \nottrans{a} \lor\; ($\\
        \tab$a \in ((U_s\cap L_e) \cup (L_s \cap U_e))\setminus L_t \implies q \nottrans{a}\;\land$\\
        \tab$a \in ((U_s\cap L_t) \cup (L_s \cap U_t))\setminus L_e \implies r \nottrans{a}\;\land$\\
        \tab$a \in (U_s\cap L_e \cap L_t) \cup (L_s \cap U_e \cap L_t) \cup (L_s \cap L_u \cap U_t) \implies (q \nottrans{a}\lor\; r \nottrans{a})$\\
        \proofstep{\cref{def:parcomp}:$\parcomp$}\\
        $\forall a \in (U_s \setminus (L_e \cup L_t)) \cup \{\tau\}: p \nottrans{a} \land$\\
        $\forall a \in  (U_e \setminus (L_s \cup L_t)) \cup \{\tau\}: q \nottrans{a}\;\land$\\
        $\forall a \in  (U_t \setminus (L_s \cup L_e)) \cup \{\tau\}: r \nottrans{a}\;\land$\\
        $\forall a \in  ((U_e \cap L_t) \cup (U_t \cap L_e)) \setminus L_s: q \nottrans{a} \lor\; r \nottrans{a}\;\land$\\
        $\forall a \in (U_s \cap (L_e \cup L_t))\cup(L_s \cap (U_e \cup U_t)): p \nottrans{a} \lor\; q\parcomp r \nottrans{a}\land$\\
        \proofstep{\cref{def:parcomp}: $\parcomp$}\\
        $\forall a \in (U_s \setminus (L_e \cup L_t)) \cup \{\tau\}: p \nottrans{a} \land$\\
        $\forall a \in ((U_e \cup U_t) \setminus L_s) \cup \{\tau\}: q \parcomp r \nottrans{a} \land$\\
        $\forall a \in (U_s \cap (L_e \cup L_t))\cup(L_s \cap (U_e \cup U_t)): p \nottrans{a} \lor\; q\parcomp r \nottrans{a}\land$\\
        \proofstep{\cref{def:parcomp}: $U_{e\parcomp t}$ and $L_{e\parcomp t}$}\\
        $\forall a \in (U_s \setminus L_{e\parcomp t}) \cup \{\tau\}: p \nottrans{a} \land$\\
        $\forall a \in (U_{e\parcomp t} \setminus L_s) \cup \{\tau\}: q \parcomp r \nottrans{a} \land$\\
        $\forall a \in (U_s \cap L_{e\parcomp t})\cup(L_s \cap U_{e\parcomp t}): p \nottrans{a} \lor\; q\parcomp r \nottrans{a}\land$\\
        \proofstep{\cref{def:parcomp}:$\parcomp$}\\
        $\forall a \in U_{s\parcomp(e\parcomp t)}:p\parcomp (q\parcomp r) \nottrans{a} $\\
        \proofstep{\cref{def:delta}:$\delta$}\\
        $p\parcomp (q\parcomp r) \trans{\delta} p\parcomp (q\parcomp r)$\\
        \proofstep{Apply subproof: $p = p' \land q = q' \land r=r'$}\\
        $p\parcomp (q\parcomp r) \trans{\delta} p'\parcomp (q'\parcomp r')$\\
        
    \end{case_distinction}

    \end{case_distinction}

\end{proof}

\noindent
Our definition for \emph{composable} is weaker than the one in other papers: $I_s \cap I_e = U_s \cap U_e = \emptyset$ 
\cite{daca_CompositionalSpecificationsIoco_2014,vanderbijl_CompositionalTestingIoco_2004,benes_CompleteCompositionOperators_2015}. This is because requiring disjoint input sets leads to a composition operator that is not associative \cite{berendsen_ParallelCompositionPaper_2008}.
A more detailed discussion of the properties of various types of parallel composition can be found in \cite{vogler_LineartimeBranchingtimePerspective_2020}. With our less restrictive definition of \emph{composable}, parallel composition is both associative and commutative as expressed in \cref{lemma:parcomp_assoc}. This is important, as it means that more than two components can also be composed and the order in which components are composed does not matter. The remaining restriction of disjoint output sets does not really restrict the applicability of parallel composition. Output sets can always be made disjoint by renaming one output label and then duplicating the synchronising transitions for the new label.

Another common approach to parallel composition is to replace all synchronised transitions with $\tau$ transitions. This is done under the assumption that communication between components is by default not observable by the outside world and therefore should be hidden. A downside is that this removes information, which makes specification-based analysis less useful. Additionally, a large part of the model-based testing theory assumes convergence, i.e.\ the absence of divergence. This means that there are no infinite paths of just $\tau$-transitions possible in the specification. By automatically hiding the labels of synchronised transitions, divergence is often introduced into the composed specification. For these reasons, we choose not to automatically hide labels during composition. 

The main purpose of a labelled transition system when used for model-based testing is to describe when an implementation is considered correct. This is done through a conformance relation.

Two common conformance relations are $\ioco$ \cite{tretmans_ModelBasedTesting_2008} and the more recent $\uioco$ relation \cite{vanderbijl_CompositionalTestingIoco_2004}.
$\uioco$ differs from $\ioco$ in how it deals with \emph{nondeterministic under-specification}, i.e.\ how non-specified inputs are handled. Among others, $\uioco$ is better suited for reasoning about composition. A detailed comparison of the two relations can be found in \cite{tretmans_GoodbyeIoco_2022}.

\begin{definition}
For $s \in \LTS$, ${\delta} \notin L_s$ is a special output denoting the absence of outputs, called \emph{quiescence}. It is defined as follows (with $p_1, p_2 \in Q_s$):
$$\begin{array}[t]{l@{~~~}c@{~~~}l}
p_1 \xrightarrow{\delta} p_2 &
\defeq &
p_1 = p_2 ~\land~ \forall x \in U_s \cup \{\tau\} :~ p_1 \nottrans{x} \\
\end{array}$$
$L^\delta$, $U^\delta$ is used as shorthand for $L \cup \{\delta\}$, $U \cup \{\delta\}$ respectively.
\label{def:delta}
\end{definition}

\begin{definition}
Let $s \in \LTS$; $p_1 \in Q_s$; $P\subseteq Q_s$ and $\sigma\in {L_s^\delta}^*$.
$$\begin{array}[t]{l@{~~}c@{~~}l}
    p_1 \after \sigma & \defeq &
    \{\ p_2 \in Q_s \ \setbar\ p_1\xRightarrow{\sigma}p_2\ \} \\
    \outset{p_1} & \defeq & 
    \{\ x \in U_s^\delta \ \setbar\ p_1 \xrightarrow{x}\ \} \\
    \outset{P} & \defeq &
    \bigcup\ \{\ \outset{p} \ \setbar\ p\in P\ \} \\
\end{array}$$
\label{def:out}
\label{def:after}
\end{definition}

\begin{definition}
Let $i \in \IOTS[I,U]$; $s\in \LTS[I,U]$:
$$\begin{array}[t]{l@{~~}l@{~~}ll}
    \utraces[s] & \defeq & \{\ \sigma \in {L^\delta}^*\ \setbar\ s \xRightarrow{\sigma} \:\land\ 
    (\ & \nexists p \in Q_s,\:\sigma_1\cdot{a}\cdot\sigma_2 = \sigma : \\
       & & & a \in I ~\land~ s\xRightarrow{\sigma_1}p ~\land~ p\Nottrans{a}\ )\ \} \\
    i \uioco s & \defeq & \multicolumn{2}{l}{\forall\sigma \in \utraces[s] :~ \outset{i \after \sigma} \subseteq \outset{s \after \sigma}} \\
\end{array}$$
\label{def:uioco}
\end{definition}

\section{Motivating Example: A Parking System}
\label{sec:exampleIntro}

We argue that parallel composition does not work nicely with $\uioco$, which we will show with an example in this section. Consider two components that together function as an automatic parking system in a car: a sensor which observes the environment and an actuator that parks the car. An illustration of how these two components communicate with each other and their environment is shown in \cref{fig:car_schematic}. Specifications for the behaviour of these components are shown in solid black in \cref{fig:carParkSensorComponents}. Their behaviour is straightforward: the parking component keeps parking as long as the sensor tells it that it is safe to do so, but stops parking if there is an obstacle, at which point it will stop the car and turn the sensor off. These components are left under-specified on purpose: it does not really matter what the sensor does if it detects an obstacle after it has been turned off, as long as it does not start beeping. This gives an implementer of the actual sensor some freedom, but still specifies the important behaviour.

\begin{figure}
    \centering
    \includegraphics[width=\linewidth]{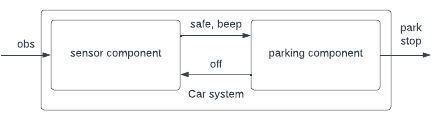}
    \caption{Two component parking system}
    \label{fig:car_schematic}
\end{figure}

Possible implementations that are $\uioco$ correct are also given in \cref{fig:carParkSensorComponents} using the extra dashed blue transitions. On first glance this all seems to make sense, and model-based testing will not find any problems when testing the components. I.E. $\cref{imp:sensor} \uioco \cref{spec:sensor} \land \cref{imp:auto_park} \uioco \cref{spec:auto_park}$. After composing our components using parallel composition, however, which is shown in \cref{fig:carParkSensorComposed}, the composed implementation is not $\uioco$ correct to the composed specification. 

The problem with the implementation in \cref{fig:carParkSensorComposed} is that it contains unspecified output transitions. These can be seen as some of the dashed transitions, which are only present in the implementation and not in the specification. This means that the previously valid implementations are now generating outputs that are not part of the composed specification. Model-based testing will report an error here, while the components are actually behaving as specified. Additionally, hidden within these false positives, there is also an actual error: if the sensor detects an obstacle after already having communicated that there is no obstacle, the parking system will not respond and will just continue parking. This is represented by the $!beep$ transition from B3 to B1, which could for instance happen if a moving obstacle like a person is present. This shows that only looking at the individual components is not enough, as there are real problems that only become visible when looking at combinations of components together.

We argue that the main problem with this example is that the component specifications rely on unspecified behaviour. The sensor specification describes exactly when the sensor is allowed to beep, but the parking specification does not always specify what the result should be. There is no guarantee that the result does not crash the system or violate any requirements. One way this could be resolved is by expanding the specifications to be input complete \cite{vanderbijl_CompositionalTestingIoco_2004}. However, doing so would remove the possibility for under-specification, which is a desirable feature in modelling behaviour. Under-specification keeps models smaller and more readable, and gives more freedom when implementing the specification. Another approach is therefore desired: a specification should specify all the behaviour that is used by other specifications, but leave the possibility of not specifying unused behaviour. This goal will be made more concrete in \cref{sec:definitions}.

\begin{figure}[t!]
\begin{subfigure}[b]{.40\linewidth}
\phantomcaption\label{subfig:carSensor}
\speclabel{sensor}
\implabel{sensor}
\begin{tikzpicture}[LTS]
\node[state,initial] (1) {1};
\node[state, below right= .5\nodedistance and \nodedistance of 1] (2) {2};
\node[state, below=of 1] (3) {3};

\path[->] 
    (1) edge [bend left] node [auto] {$\mathit{?off}$} (2)
        edge [] node [auto] {$?obs$} (3)
        edge [loop above] node [below left=-2mm and 2mm] {$\mathit{!safe}$} (1)
    (2) edge [loop right, imp only] node [below left=2mm and -2mm] {$?obs$\\$\mathit{?off}$} (2)
    (3) edge [bend right] node [auto] {$\mathit{?off}$} (2)
        edge [bend left] node [auto] {$!beep$} (1)
    
        edge [loop left, imp only] node [auto] {$?obs$} (3);

\end{tikzpicture}
\captiontext{\currentspec{} and \currentimp: car sensor component}
\end{subfigure}
\begin{subfigure}[b]{.59\linewidth}
\phantomcaption\label{subfig:carauto_park}
\speclabel{auto_park}
\implabel{auto_park}
\begin{tikzpicture}[LTS]
\node[initial,state] (A) {A};
\node[state, below=of A] (B) {B};
\node[state, below right=.5\nodedistance and \nodedistance of A] (C) {C};
\node[state, above right=.5\nodedistance and \nodedistance of C] (D) {D};
\node[state, below=of D] (E) {E};

\path[->] 
    (A) edge [] node [left] {$\mathit{?safe}$} (B)
        edge [] node [auto] {$?beep$} (C)
    (B) edge [bend right] node [right] {$!park$} (A)
        edge [loop right,imp only] node [auto] {$\mathit{?safe}$\\$?beep$} (B)
    (C) edge [] node [auto] {$!stop$} (D)
        edge [loop below right,imp only] node [below] {$\mathit{?safe}$\\$?beep$} (C)
    (D) edge [] node [auto] {$\mathit{!off}$} (E)
        edge [loop right,imp only] node [auto] {$\mathit{?safe}$\\$?beep$} (D)
    (E) edge [loop right, imp only] node [auto] {$\mathit{?safe}$\\$?beep$} (E);%
\end{tikzpicture}
\captiontext{\currentspec{} and \currentimp{}: automated parking component}%
\end{subfigure}
\caption{Car component specifications ($\tikzarrow$) and implementations ($\tikzarrow[imp only]$)}
\label{fig:carParkSensorComponents}
\end{figure}
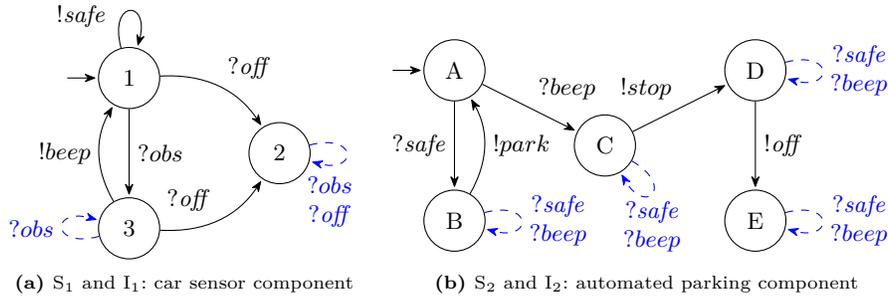%
\begin{figure}
\centering
\begin{tikzpicture}[LTS]
\node[initial,state] (A1) {A1};
\node[state, right of =  A1] (B1) {B1};
\node[state, below of= A1] (A3) {A3};
\node[state, right of= A3] (B3) {B3};
\node[state, right of= B3] (C1) {C1};
\node[state, above of= C1] (C3) {C3};
\node[state, right of= C1] (D1) {D1};
\node[state, right of= C3] (D3) {D3};
\node[state, below right = 0.5\nodedistance and \nodedistance of D3] (E2) {E2};

\path[->] 
    (A1)    edge[bend left] node [above] {$\mathit{!safe}$} (B1)
            edge node [left] {$?obs$} (A3)
    (B1)    edge[bend left] node [below] {$!park$} (A1)
            edge[bend left] node [below right] {$?obs$} (B3)
            edge[loop above, imp only] node [above] {$\mathit{!safe}$} (B1) 
    (B3)    edge node [above] {$!park$} (A3)
            edge [bend left,imp only] node [below left=-1pt and 1pt] {$!beep$} (B1)
            edge [loop below, imp only] node [below] {$?obs$} (B3)
    (A3)    edge [bend right=50, looseness=1.3] node [below]
                {$!beep$} (C1)
            edge [loop below,imp only] node[below] {$?obs$} (A3)
    (C1)    edge [bend left] node [above left] {$?obs$} (C3)
            edge node [below] {$!stop$} (D1)
            edge[loop below,imp only] node [below] {$\mathit{!safe}$} (C1)
    (C3)    edge[bend left,imp only] node [below right] {$!beep$}
                (C1)
            edge node [above] {$!stop$} (D3)
            edge [loop above, imp only] node [above] {$?obs$} (C3)
    (D1)    edge [bend left] node [above left] {$?obs$}
                (D3)
            edge node [below right] {$\mathit{!off}$} (E2)
            edge[loop below,imp only] node [below] {$\mathit{!safe}$} (D1)
    (D3)    edge [bend left,imp only] node [right] {$!beep$} (D1)
            edge node [above right] {$\mathit{!off}$} (E2)
            edge [loop above, imp only] node [above] {$?obs$} (D3)
    (E2)    edge [loop right, imp only] node [right] {$?obs$} (E2)
    ;

\end{tikzpicture}
\composedlabel{spec:sensor}{spec:auto_park}{spec:composed_park}
\composedlabel{imp:sensor}{imp:auto_park}{imp:composed_park}
\caption{Car autopark and sensor composed \cref{spec:sensor}$\parcomp$\cref{spec:auto_park} ($\tikzarrow$) and \cref{imp:sensor}$\parcomp$\cref{imp:auto_park} ($\tikzarrow[imp only]$)} 
\label{fig:carParkSensorComposed}
\end{figure}

\section{Mutual Acceptance}
\label{sec:definitions}
In order to reason about specified and unspecified behaviour an explicit notion of what it means for behaviour to be specified is required. For $\uioco$, the allowance of outputs is always explicitly specified. They are either present in the model and therefore allowed, or absent and disallowed. After a specified input, the model again defines what is allowed. Inputs are always implicitly allowed, but if an input is not part of the model all behaviour after that input is allowed. This means that the behaviour after an absent input is unspecified: the model does not tell us what should or should not happen. Therefore, if all outputs given by one component, are inputs present in the model of the other component, there will be no unspecified behaviour. This requirement is formulated in \cref{def:inset,def:projection,def:accepting}: if
after some $\sigma \in \utraces[s\parcomp e]$ some pair of states $s',e'$ is reached, and $s'$ produces a synchronised output, then $e'$ must have this output as an input. Note that this is trivially holds if $e$ is input enabled which generalises earlier results about component-based testing with $\uioco$ \cite{vanderbijl_CompositionalTestingIoco_2004}.

\begin{definition}
For $s\in LTS$; $p \in Q_s$; $P \subseteq Q_s$, the set of enabled inputs is defined as:
\centermath{\begin{array}[t]{l@{~~}l@{~~}l}
    \inset{p} & \defeq & \{ \ \ell \in I_s \setbar p\xRightarrow{\ell}\ \} \\
    \inset{P} & \defeq & \bigcap\,\{\ \inset{q} \setbar q\in P\ \}
\end{array}}
\label{def:inset}
\end{definition}

\begin{definition}
Let $\sigma \in L^{\delta*}$; $\mathcal{L}\subseteq L^\delta$; and $\ell \in L^\delta$.
Projecting a trace to a smaller set of labels is defined as:
\centermath{\begin{array}[t]{l@{~~}l@{~~}l@{~}l}
    \project{\epsilon}{\mathcal{L}} & \defeq & \epsilon & \\
    \project{(\sigma\cdot\ell)}{\mathcal{L}} & \defeq &
    (\project{\sigma}{\mathcal{L}})\cdot\ell & \text{\ if\ \ } \ell \in \mathcal{L} \\
    & & \phantom{(}\project{\sigma}{\mathcal{L}} & \text{\ otherwise}
\end{array}}
\label{def:projection}
\end{definition}

\begin{definition}
Let $s,e \in \LTS$ be $\composable$, then $s$ \textbf{accepts} $e$ iff:
\[\begin{array}[t]{l@{~~~}l@{~~~}l}
    s \accepts e & \defeq &
    \forall \sigma \in \utraces[s\parcomp e],\ s' \in Q_s,\ e' \in Q_e: \\
    & & ~~~~~~~ s \parcomp e \xRightarrow{\sigma} s' \parcomp e' \ \implies\ 
    \outset{e'}\cap I_s \:\subseteq\: \inset{s'} \cap U_e
\end{array}\]
\label{def:accepting}    
\end{definition}

The symmetric version of the $\accepts$ relation is defined in \cref{def:mutually_accepts}. Though it might look like an equivalence relation, it is neither reflexive nor transitive. Reflexivity fails because $\mutuallyaccepts$ is indirectly defined using parallel composition. This means it is only defined on specifications that are composable, and any specification with outputs is not composable with itself. Transitivity is also not true, because each pair of specifications has its own sets of state pairs and shared labels for which the $\accepts$ relation must hold. This means each specification pair must be checked independently of any other specifications.

\begin{definition}
Let $s,e \in \LTS$ be $\composable$, then $s$ \textbf{mutually accepts} $e$:
$$s \mutuallyaccepts e ~~ \defeq ~~ s \accepts e \ \land\ e \accepts s$$
\label{def:mutually_accepts}
\end{definition}

\section{Compositionalility for uioco}
\label{sec:acceptingSystems}

The previous section defined what it means for a specification to not trigger undefined behaviour in another specification using the $\accepts$ relation. This section will prove that this property  allows compositional testing using $\uioco$.

\begin{toappendix}
\Cref{lemma:parcomp_base_properties,lemma:project_from_parcomp,lemma:IOTS_step_both} are based on similar lemmas from \cite{vanderbijl_CompositionalTestingIoco_2004}. The lemmas and proofs are adapted to no longer require input enabled specifications, and use the $\mutuallyaccepts$ relation instead. Additionally, our definition for $\composable$ is also different.
\end{toappendix}

\begin{toappendix}
\begin{lemma}
\label{lemma:parcomp_base_properties}
let $s,e \in \LTS$ be $\composable$, $p_1, p_2\in Q_s, q_1, q_2\in Q_e$

\begin{enumerate}
    \item \label{item:parcomp_base_properties_step_left}
    \[\forall \ell \in L_s \setminus L_e : p_1 \parcomp q_1 \xrightarrow{\ell} p_2 \parcomp q_1 \iff p_1 \xrightarrow{\ell} p_2\]

    \item \label{item:parcomp_base_properties_step_right}
    \[\forall \ell \in L_e \setminus L_s : p_1 \parcomp q_1 \xrightarrow{\ell} p_1 \parcomp q_2 \iff q_1 \xrightarrow{\ell} q_2\]

    \item \label{item:parcomp_base_properties_step_tau}
    \[p_1 \parcomp q_1 \xrightarrow{\tau} p_2 \parcomp q_2 \iff (p_1 \xrightarrow{\tau} p_2 \land q_1 = q_2)\lor (q_1 \xrightarrow{\tau} q_2 \land p_1 = p_2)\]

    \item \label{item:parcomp_base_properties_step_both}
    \begin{align*}
        \forall \ell \in (L_s \cap L_e) \cup \{\delta\} : p_1 \parcomp q_1 \xrightarrow{\ell} p_2 \parcomp q_2\, \land \\
        s \mutuallyaccepts e \,\land \\
        (\exists \sigma \in \utraces[s \parcomp e]: s \parcomp e \xRightarrow{\sigma} p_1 \parcomp q_1) \implies \\ 
        p_1 \xrightarrow{\ell} p_2 \land q_1 \xrightarrow{\ell} q_2
    \end{align*}

    \item \label{item:parcomp_base_properties_both_to_parcomp}
    \[\forall \ell \in (L_s \cap L_e) \cup \{\delta\} : p_1 \xrightarrow{\ell} p_2 \land q_1 \xrightarrow{\ell} q_2  \implies p_1 \parcomp q_1 \xrightarrow{\ell} p_2 \parcomp q_2\]

\end{enumerate}
\end{lemma}
\begin{proof}
~\begin{enumerate}

    \item Take $a\in L_s \setminus L_e$
    \begin{case_distinction}
        \item[$\Longrightarrow$:]\ \\
        $p_1 \parcomp q_1 \xrightarrow{a} p_2 \parcomp q_2$\\
        \proofstep{Only possible through first case of definition $\parcomp_T$, \cref{def:parcomp}}\\
        $p_1 \xrightarrow{a} p_2$
        \item[$\Longleftarrow$:]\ \\
        $p_1 \xrightarrow{a} p_2$\\
        \proofstep{First case of definition $\parcomp_T$, \cref{def:parcomp}}\\
        $\forall q \in Q_e: p_1 \parcomp q \xrightarrow{a} p_2 \parcomp q$
    \end{case_distinction}
    
    \item Analogous to 1.
    
    \item Analogous to 1.
    
    \item The case $a \in L_s \cap L_e$ is analogous to 1. The other case of $a = \delta$ is given below.\\ 
    Take $\sigma \in \utraces[s\parcomp e]$\\
    
    $p_1 \parcomp q_1 \xrightarrow{\delta} p_2 \parcomp q_2 \land s \mutuallyaccepts e \land s \parcomp e \xRightarrow{\sigma} p_1 \parcomp q_1$\\
    \proofstep{\Cref{def:delta}: $\delta$}\\
    $p_1 \parcomp q_1 \xrightarrow{\delta} p_2 \parcomp q_2 \land s \mutuallyaccepts e \land s \parcomp e \xRightarrow{\sigma} p_1 \parcomp q_1 \land p_1 = p_2 \land q_1 = q_2$\\
    \proofstep{\Cref{def:delta}: $\delta$ + remember $p_1 = p_2 \land q_1 = q_2$ as subproof, dropped here for brevity }\\
    $\forall a \in U_s \cup U_e \cup \{\tau\}: p_1 \parcomp q_1 \nottrans{a} \land s \mutuallyaccepts e \land s \parcomp e \xRightarrow{\sigma} p_1 \parcomp q_1$\\
    \proofstep{\Cref{def:parcomp}: $\parcomp$}\\
    $\forall a \in (U_s \setminus L_e) \cup \{\tau\}: p_1 \nottrans{a} \land$\\
    $\forall a \in (U_e \setminus L_s) \cup \{\tau\}: q_1 \nottrans{a} \land$\\
    $\forall a \in (U_s \cap L_e) \cup (U_e \cap L_s): p_1 \nottrans{a} \lor\; q_1 \nottrans{a}\land$\\
    $s \mutuallyaccepts e\; \land$\\
    $s \parcomp e \xRightarrow{\sigma} p_1 \parcomp q_1$\\
    \proofstep{Apply \cref{def:mutually_accepts,def:accepting}: mutually accepts and accepts}\\
    $\forall a \in (U_s \setminus L_e) \cup \{\tau\}: p_1 \nottrans{a} \land$\\
    $\forall a \in (U_e \setminus L_s) \cup \{\tau\}: q_1 \nottrans{a} \land$\\
    $\forall a \in (U_s \cap L_e) \cup (U_e \cap L_s): p_1 \nottrans{a} \lor\; q_1 \nottrans{a}\land$\\
    $(\forall \sigma' \in \utraces[s \parcomp e], p \in Q_s, q \in Q_e: s \parcomp e \xRightarrow{\sigma'} p \parcomp q \implies$\\
    $out(p) \cap I_e \subseteq in(q) \cap U_s\; \land$\\ 
    $out(q) \cap I_s \subseteq in(p) \cap U_e )\; \land $\\
    $s \parcomp e \xRightarrow{\sigma} p_1 \parcomp q_1$\\
    \proofstep{$\forall$ elimination + $\implies$ elimination}\\
     $\forall a \in (U_s \setminus L_e) \cup \{\tau\}: p_1 \nottrans{a} \land$\\
    $\forall a \in (U_e \setminus L_s) \cup \{\tau\}: q_1 \nottrans{a} \land$\\
    $\forall a \in(U_s \cap L_e) \cup (U_e \cap L_s): p_1 \nottrans{a} \lor\; q_1 \nottrans{a} \land$\\
    $ out(p_1) \cap I_e \subseteq in(q_1) \cap U_s \;\land$\\ 
    $ out(q_1) \cap I_s \subseteq in(p_1) \cap U_e$ \\
    \proofstep{\Cref{def:out,def:inset}: $out$ and $in$}\\
    $\forall a \in (U_s \setminus L_e) \cup \{\tau\}: p_1 \nottrans{a} \land$\\
    $\forall a \in (U_e \setminus L_s) \cup \{\tau\}: q_1 \nottrans{a} \land$\\
    $\forall a \in(U_s \cap L_e) \cup (U_e \cap L_s): p_1 \nottrans{a} \lor\; q_1 \nottrans{a}\land$\\
    $\forall a \in U_s \cap I_e: p_1 \xrightarrow{a} \implies  q_1 \xRightarrow{a} \land$\\ 
    $\forall a \in U_e \cap I_s: q_1 \xrightarrow{a} \implies  p_1 \xRightarrow{a} \land$\\
    \proofstep{$p_1 \nottrans{\tau} \land\; q_1\nottrans{\tau}$}\\
    $\forall a \in (U_s \setminus L_e) \cup \{\tau\}: p_1 \nottrans{a} \land$\\
    $\forall a \in (U_e \setminus L_s) \cup \{\tau\}: q_1 \nottrans{a} \land$\\
    $\forall a \in (U_s \cap L_e) \cup (U_e \cap L_s): p_1 \nottrans{a} \lor\; q_1 \nottrans{a}\land$\\
    $\forall a \in U_s \cap I_e : p_1 \xrightarrow{a} \implies q_1 \xrightarrow{a} \land$\\
    $\forall a \in U_e \cap I_s : q_1 \xrightarrow{a} \implies p_1 \xrightarrow{a}$\\
    \proofstep{$(A \implies B) \iff (\neg B \implies \neg A)$}\\
    $\forall a \in (U_s \setminus L_e) \cup \{\tau\}: p_1 \nottrans{a} \land$\\
    $\forall a \in (U_e \setminus L_s) \cup \{\tau\}: q_1 \nottrans{a} \land$\\
    $\forall a \in(U_s \cap L_e) \cup (U_e \cap L_s): p_1 \nottrans{a} \lor\; q_1 \nottrans{a}\land$\\
    $\forall a \in U_s \cap I_e : q_1 \nottrans{a} \implies p_1 \nottrans{a}\land$\\
    $\forall a \in U_e \cap I_s : p_1 \nottrans{a} \implies q_1 \nottrans{a}$\\
    \proofstep{\cref{def:composable}: $U_s \cap U_e = \emptyset$}\\
    $\forall a \in (U_s \setminus L_e) \cup \{\tau\}: p_1 \nottrans{a} \land$\\
    $\forall a \in (U_e \setminus L_s) \cup \{\tau\}: q_1 \nottrans{a} \land$\\
    $\forall a \in (U_s \cap L_e) \cup (U_e \cap L_s): p_1 \nottrans{a} \lor\; q_1 \nottrans{a}\land$\\
    $\forall a \in U_s \cap L_e : q_1 \nottrans{a} \implies p_1 \nottrans{a}\land$\\
    $\forall a \in U_e \cap L_s : p_1 \nottrans{a} \implies q_1 \nottrans{a}$\\
    \proofstep{$(A \lor B) \land (A \implies B) \implies B$}\\
    $\forall a \in (U_s \setminus L_e) \cup \{\tau\}: p_1 \nottrans{a} \land$\\
    $\forall a \in (U_e \setminus L_s) \cup \{\tau\}: q_1 \nottrans{a} \land$\\
    $\forall a \in U_s \cap L_e :  p_1 \nottrans{a}\land$\\
    $\forall a \in U_e \cap L_s :  q_1 \nottrans{a}$\\
    \proofstep{$(X \setminus Y) \cup (X \cap Y) = X$}\\
    $\forall a \in U_s \cup \{\tau\} : p_1 \nottrans{a} \land$\\ $\forall a \in U_e \cup \{\tau\} : q_1 \nottrans{a}$\\
    \proofstep{\Cref{def:delta}: $\delta$}\\
    $p_1 \xrightarrow{\delta} p_1 \land q_1 \xrightarrow{\delta} q_1$\\
    \proofstep{Apply subproof: $p_1 = p_2 \land q_1 = q_2$}\\
    $p_1 \xrightarrow{\delta} p_2 \land q_1 \xrightarrow{\delta} q_2$

    \item The case $a \in L_s \cap L_e$ is analogous to 1. The other case of $a = \delta$ is given below.\\
    $p_1 \xrightarrow{\delta} p_2 \land q_1 \xrightarrow{\delta} q_2$\\
    \proofstep{\Cref{def:delta}: $\delta$}\\
    $\forall \ell \in U_s \cup \{\tau\} : p_1 \nottrans{\ell} \land\; \forall \ell \in U_e \cup \{\tau\} : q_1 \nottrans{\ell} \land \; p_1 = p_2 \land q_1 = q_2$\\
    \proofstep{\Cref{def:parcomp}: $\parcomp$}\\
    $\forall \ell \in U_s \cup U_e \cup \{\tau\} : p_1 \parcomp q_1 \nottrans{\ell} \land\; p_1 = p_2 \land q_1 = q_2$\\
    \proofstep{\Cref{def:delta}: $\delta$}\\
    $p_1 \parcomp q_1 \xrightarrow{\delta} p_1 \parcomp q_1 \land p_1 = p_2 \land q_1 = q_2$\\
    \proofstep{Rewrite using equalities}\\
    $p_1 \parcomp q_1 \xrightarrow{\delta} p_2 \parcomp q_2$\\
\end{enumerate}
\end{proof}
\end{toappendix}

\begin{toappendix}
\begin{lemma}
\label{lemma:IOTS_step_both}
Let $s,e\in \IOTS$, $\ell \in (L_s \cap L_e) \cup \{\delta\}$, $p_1,p_2 \in Q_s$, $q_1,q_2 \in Q_e$
        \[\forall  : p_1 \parcomp q_1 \xrightarrow{\ell} p_2 \parcomp q_2\ \iff p_1 \xrightarrow{\ell} p_2 \land q_1 \xrightarrow{\ell} q_2\]
\end{lemma}
\begin{proof}
\
\begin{case_distinction}
        \item[($\implies$):] Proof by case distinction on $\ell$
        \begin{enumerate}[align=left]
            \item[$\ell \in L_s \cap L_e$:]\ \\
            $p_1 \parcomp q_1 \xrightarrow{\ell} p_2 \parcomp q_2$\\
            \proofstep{\Cref{def:parcomp}: $\parcomp$}\\
            $p_1 \xrightarrow{\ell} p_2 \land q_1 \xrightarrow{\ell} q_2$
            \item[$\ell = \delta$:]\ \\
            $p_1 \parcomp q_1 \xrightarrow{\delta} p_2 \parcomp q_2$\\
            \proofstep{\Cref{def:delta}: $\delta$}\\
            $\forall a \in U_s \cup U_e \cup \{\tau\}: p_1 \parcomp q_1 \nottrans{a} \land p_1 = p_2 \land q_1 = q_2$\\
            \proofstep{\Cref{def:parcomp}: $\parcomp$ + remember $p_1 = p_2 \land q_1 = q_2$ as subproof, dropped here for brevity}\\
            $\forall a \in (U_s \setminus L_e) \cup \{\tau\}: p_1 \nottrans{a} \land$\\
            $\forall a \in (U_e \setminus L_s) \cup \{\tau\}: q_1 \nottrans{a} \land$\\
            $\forall a \in L_s \cap L_e: p_1 \nottrans{a} \lor\; q_1 \nottrans{a}\land$\\
            \proofstep{\Cref{def:iots}: $IOTS$}\\
            $\forall a \in (U_s \setminus L_e) \cup \{\tau\}: p_1 \nottrans{a} \land$\\
            $\forall a \in (U_e \setminus L_s) \cup \{\tau\}: q_1 \nottrans{a} \land$\\
            $\forall a \in L_s \cap L_e: p_1 \nottrans{a} \lor\; q_1 \nottrans{a}\land$\\
            $\forall a \in I_s: p_1 \xrightarrow{a} \land\; \forall a \in I_e: q_1 \xrightarrow{a}\land$\\
            \proofstep{$I_s \cap L_e \subseteq I_s \land I_e \cap L_s \subseteq I_e$}\\
            $\forall a \in (U_s \setminus L_e) \cup \{\tau\}: p_1 \nottrans{a} \land$\\
            $\forall a \in (U_e \setminus L_s) \cup \{\tau\}: q_1 \nottrans{a} \land$\\
            $\forall a \in L_s \cap L_e: p_1 \nottrans{a} \lor\; q_1 \nottrans{a}\land$\\
            $\forall a \in I_s \cap L_e: p_1 \xrightarrow{a}\; \land$\\
            $\forall a \in I_e \cap L_s: q_1 \xrightarrow{a}\land$\\
            \proofstep{$I_s \cap L_e \subseteq L_s \cup L_e\land I_e \cap L_s \subseteq L_s \cup L_e$}\\
            $\forall a \in (U_s \setminus L_e) \cup \{\tau\}: p_1 \nottrans{a} \land$\\
            $\forall a \in (U_e \setminus L_s) \cup \{\tau\}: q_1 \nottrans{a} \land$\\
            $\forall a \in I_s \cap L_e: p_1 \nottrans{a} \lor\; q_1 \nottrans{a}\land$\\
            $\forall a \in I_e \cap L_s: p_1 \nottrans{a} \lor\; q_1 \nottrans{a}\land$\\
            $\forall a \in I_s \cap L_e: p_1 \xrightarrow{a}\; \land$\\
            $\forall a \in I_e \cap L_s: q_1 \xrightarrow{a}\land$\\
            \proofstep{$(\neg A \lor \neg B) \land A \implies \neg B $}\\
            $\forall a \in (U_s \setminus L_e) \cup \{\tau\}: p_1 \nottrans{a} \land$\\
            $\forall a \in (U_e \setminus L_s) \cup \{\tau\}: q_1 \nottrans{a} \land$\\
            $\forall a \in I_e \cap L_s :  p_1 \nottrans{a}\land$\\
            $\forall a \in I_s \cap L_e :  q_1 \nottrans{a}$\\
            \proofstep{$I_e \cap U_s \subseteq I_e \cup L_s\land I_s \cap U_e \subseteq I_s \cup L_e$}\\
            $\forall a \in (U_s \setminus L_e) \cup \{\tau\}: p_1 \nottrans{a} \land$\\
            $\forall a \in (U_e \setminus L_s) \cup \{\tau\}: q_1 \nottrans{a} \land$\\
            $\forall a \in  I_e \cap U_s :  p_1 \nottrans{a}\land$\\
            $\forall a \in  I_s \cap U_e  :  q_1 \nottrans{a}$\\
            \proofstep{\cref{def:composable}: $U_s \cap U_e = \emptyset$}\\
            $\forall a \in (U_s \setminus L_e) \cup \{\tau\}: p_1 \nottrans{a} \land$\\
            $\forall a \in (U_e \setminus L_s) \cup \{\tau\}: q_1 \nottrans{a} \land$\\
            $\forall a \in  L_e \cap U_s :  p_1 \nottrans{a}\land$\\
            $\forall a \in  L_s \cap U_e  :  q_1 \nottrans{a}$\\
            \proofstep{$(X \setminus Y) \cup (X \cap Y) = X$}\\
            $\forall a \in U_s \cup \{\tau\} : p_1 \nottrans{a} \land$\\ $\forall a \in U_e \cup \{\tau\} : q_1 \nottrans{a}$\\
            \proofstep{Apply subproof: $p_1 = p_2 \land q_1 = q_2$}\\
            $\forall a \in U_s \cup \{\tau\} : p_1 \nottrans{a} \land$\\ $\forall a \in U_e \cup \{\tau\} : q_1 \nottrans{a}$\\
            $p_1 = p_2 \land q_1 = q_2$\\
            \proofstep{\Cref{def:delta}: $\delta$}\\
            $p_1 \xrightarrow{\delta} p_2 \land q_1 \xrightarrow{\delta} q_2$
        \end{enumerate}
        \item[($\impliedby$):] Covered by \Itemref{lemma:parcomp_base_properties}{both_to_parcomp}
\end{case_distinction}
\end{proof}

\Cref{lemma:parcomp_base_properties} states the relationships between being able to do a single transition in a composed system, and being able to do that same transition in the component systems. \Itemref{lemma:parcomp_base_properties}{both_to_parcomp} does not hold in the other direction in general for labelled transition systems, but it does work both ways for input enabled transition systems which is shown in \cref{lemma:IOTS_step_both}. The other direction is also proven for regular labeled transition systems in \itemref{lemma:parcomp_base_properties}{step_both}, but requires extra preconditions.
\end{toappendix}

\begin{toappendix}
\begin{proposition}
Let $s\in \LTS, \sigma,\sigma_1, \sigma_2 \in (L_s^\delta)^*$, $\ell \in L_s^\delta$, $p_1,p_2 \in Q_s$
    \begin{enumerate}
        \item \label{item:general_properties_prefix_closed}
        $\utraces$ is prefix closed:
        \[\sigma\cdot \ell \in \utraces[S] \implies \sigma \in \utraces[S]\]
        \item \label{item:general_properties_trans_transitive}
        The transition relation $\xRightarrow{}$ is transitive:
        \[(\exists p_3 \in Q_s: p_1 \xRightarrow{\sigma_1} p_3 \land p_3 \xRightarrow{\sigma_2} p_2) \iff p_1 \xRightarrow{\sigma_1\cdot\sigma_2} p_2\]
    \end{enumerate}
    \label{prop:general_properties}
\end{proposition}

\Cref{prop:general_properties} shows some general properties that are used in the remainder of this paper. All of them follow directly from the definitions.
\end{toappendix}

\begin{toappendix}  
\begin{lemmarep}
let $s,e$ be $\composable$ $LTS$, $p,p'\in Q_s$, $q,q' \in Q_e, \sigma \in  {L_{s\parcomp e}^{\delta}}^{*}$
\[p \xRightarrow{\project{\sigma}{L_s^\delta}} p' \land q \xRightarrow{\project{\sigma}{L_e^\delta}} q' \implies p \parcomp q \xRightarrow{\sigma} p' \parcomp q'\]
\label{lemma:project_from_parcomp_light}
\end{lemmarep}
\begin{proof}
Proof by induction on $\sigma$.
    \begin{case_distinction}
        \item[Base case:]\ \\
        $\sigma = \epsilon$. Let $\sigma_\tau,\sigma_\tau', \sigma_\tau'' \in \tau^*$\\
        $p \xRightarrow{\project{\epsilon}{L_s^\delta}} p' \land q \xRightarrow{\project{\epsilon}{L_e^\delta}}q'$\\
        \proofstep{\Cref{def:projection}: $\projectop$}\\
        $p \xRightarrow{\epsilon} p' \land q \xRightarrow{\epsilon}q'$\\
        \proofstep{\Cref{def:arrowdefs}: $\xRightarrow{\epsilon}$}\\
        $p \xrightarrow{\sigma_\tau'} p' \land q \xrightarrow{\sigma_\tau''}q'$\\
        \proofstep{\Itemref{lemma:parcomp_base_properties}{step_tau}}\\
        $p \parcomp q \xrightarrow{\sigma_\tau} p' \parcomp q'$\\
        \proofstep{\Cref{def:arrowdefs}: $\xRightarrow{\epsilon}$}\\
        $p \parcomp q \xRightarrow{\epsilon} p' \parcomp q'$\\
        
        \item[Induction step:]  Assume the proposition holds for $\sigma'\in  {L_{s\parcomp e}^{\delta}}^{*}$.\\
        To proof: the proposition holds for $\sigma$, where $\sigma =\sigma' \cdot a$ with $a \in L_{s \parcomp e}^\delta$. This is divided into three cases based on $a$:
        \begin{case_distinction}
            \item[$a \in L_s \setminus L_e$:]\ \\
            $p \xRightarrow{\project{\sigma'\cdot a}{L_s^\delta}} p' \land q \xRightarrow{\project{\sigma'\cdot a}{L_e^\delta}} q' $\\
            \proofstep{\Cref{def:projection}: $\projectop$}\\
            $p \xRightarrow{\project{\sigma'}{L_s^\delta}\cdot a} p' \land q \xRightarrow{\project{\sigma'}{L_e^\delta}} q' $\\
            \proofstep{\Cref{def:arrowdefs}: $\xRightarrow{\project{\sigma'}{L_s^\delta}\cdot a}$}\\
            $\exists p_1, p_2 \in Q_s: p \xRightarrow{\project{\sigma'}{L_s^\delta}} p_1 \land p_1 \xrightarrow{a} p_2 \land p_2 \xRightarrow{\epsilon} p' \land q \xRightarrow{\project{\sigma'}{L_e^\delta}} q' $\\
            \proofstep{\Cref{def:arrowdefs}: $\xRightarrow{\epsilon}$}\\
            $\exists p_1, p_2 \in Q_s: p \xRightarrow{\project{\sigma'}{L_s^\delta}} p_1 \land p_1 \xrightarrow{a} p_2 \land p_2 \xRightarrow{\epsilon} p' \land q' \xRightarrow{\epsilon} q' \land q \xRightarrow{\project{\sigma'}{L_e^\delta}} q' $\\
            \proofstep{Apply basic step}\\
            $\exists p_1, p_2 \in Q_s: p \xRightarrow{\project{\sigma'}{L_s^\delta}} p_1 \land p_1 \xrightarrow{a} p_2 \land p_2 \parcomp q' \xRightarrow{\epsilon} p' \parcomp q' \land q \xRightarrow{\project{\sigma'}{L_e^\delta}} q' $\\
            \proofstep{\Itemref{lemma:parcomp_base_properties}{step_left}}\\
            $\exists p_1, p_2 \in Q_s: p \xRightarrow{\project{\sigma'}{L_s^\delta}} p_1 \land p_1 \parcomp q' \xrightarrow{a} p_2 \parcomp q' \land p_2 \parcomp q' \xRightarrow{\epsilon} p' \parcomp q' \land q \xRightarrow{\project{\sigma'}{L_e^\delta}} q' $\\
            \proofstep{\Cref{def:arrowdefs}: $\xRightarrow{a}$}\\
            $\exists p_1 \in Q_s: p \xRightarrow{\project{\sigma'}{L_s^\delta}} p_1 \land p_1 \parcomp q' \xRightarrow{a} p' \parcomp q' \land q \xRightarrow{\project{\sigma'}{L_e^\delta}} q' $\\
            \proofstep{Apply IH}\\
            $\exists p_1 \in Q_s: p \parcomp q \xRightarrow{\sigma'} p_1 \parcomp q' \land p_1 \parcomp q' \xRightarrow{a} p' \parcomp q' $\\
            \proofstep{\Itemref{prop:general_properties}{trans_transitive}}\\
            $p \parcomp q \xRightarrow{\sigma' \cdot a} p' \parcomp  q' $\\

            \item[$a \in L_e \setminus L_s$:] Symmetric with the previous case.\\

            \item[$a\in (L_s \cap L_e) \cup \{\delta\}$:]\ \\
             $p \xRightarrow{\project{\sigma'\cdot a}{L_s^\delta}} p' \land q \xRightarrow{\project{\sigma'\cdot a}{L_e^\delta}} q' $\\
            \proofstep{\Cref{def:projection}: $\projectop$}\\
            $p \xRightarrow{\project{\sigma'}{L_s^\delta}\cdot a} p' \land q \xRightarrow{\project{\sigma'}{L_e^\delta}\cdot a} q' $\\
            \proofstep{\Cref{def:arrowdefs}: $\xRightarrow{\project{\sigma'}{L_s^\delta}\cdot a}$}\\
            $\exists p_1, p_2 \in Q_s: p \xRightarrow{\project{\sigma'}{L_s^\delta}} p_1 \land p_1 \xrightarrow{a} p_2 \land p_2 \xRightarrow{\epsilon} p' \land q \xRightarrow{\project{\sigma'}{L_e^\delta}\cdot a} q' $\\
            \proofstep{\Cref{def:arrowdefs}: $\xRightarrow{\project{\sigma'}{L_e^\delta}\cdot a}$}\\
            $\exists p_1, p_2 \in Q_s, q_1, q_2 \in Q_e:$\\ 
            $p \xRightarrow{\project{\sigma'}{L_s^\delta}} p_1 \land 
            p_1 \xrightarrow{a} p_2 \land 
            p_2 \xRightarrow{\epsilon} p' \land 
            q \xRightarrow{\project{\sigma'}{L_e^\delta}} q_1 \land 
            q_1 \xrightarrow{a} q_2 \land
            q_2 \xRightarrow{\epsilon} q' $\\
            \proofstep{Apply basic step}\\
            $\exists p_1, p_2 \in Q_s, q_1, q_2 \in Q_e:$\\ 
            $p \xRightarrow{\project{\sigma'}{L_s^\delta}} p_1 \land 
            p_1 \xrightarrow{a} p_2 \land 
            q \xRightarrow{\project{\sigma'}{L_e^\delta}} q_1 \land 
            q_1 \xrightarrow{a} q_2 \land
            p_2 \parcomp q_2 \xRightarrow{\epsilon} p' \parcomp q' $\\
            \proofstep{Apply IH}\\
            $\exists p_1, p_2 \in Q_s, q_1, q_2 \in Q_e:  
            p \parcomp q \xRightarrow{\sigma'} p_1 \parcomp q_1 \land p_1 \xrightarrow{a} p_2 \land q_1 \xrightarrow{a} q_2 \land
            p_2 \parcomp q_2 \xRightarrow{\epsilon} p' \parcomp q' $\\
            \proofstep{\Itemref{lemma:parcomp_base_properties}{both_to_parcomp}}\\
            $\exists p_1, p_2 \in Q_s, q_1, q_2 \in Q_e: 
            p \parcomp q \xRightarrow{\sigma'} p_1 \parcomp q_1 \land 
            p_1 \parcomp q_1 \xrightarrow{a} p_2 \parcomp q_2 \land 
            p_2 \parcomp q_2 \xRightarrow{\epsilon} p' \parcomp  q' $\\
            \proofstep{\Cref{def:arrowdefs}: $\xRightarrow{a}$}\\
            $\exists p_1 \in Q_s, q_1 \in Q_e:
            p \parcomp q \xRightarrow{\sigma'} p_1 \parcomp q_1 \land
             p_1 \parcomp q_1 \xRightarrow{a} p' \parcomp q'$\\
            \proofstep{\Itemref{prop:general_properties}{trans_transitive}}\\
            $p \parcomp q \xRightarrow{\sigma'\cdot a} p' \parcomp  q' $\\
        \end{case_distinction}
    \end{case_distinction}
\end{proof}
\end{toappendix}

\Cref{lemma:project_from_parcomp_IOTS} shows how for composable, input complete systems, traces in the composed system can be transformed into traces in the component systems, and the other way around. This allows for compositional model-based testing in input complete systems. 
\Cref{lemma:project_from_parcomp} then goes on to show that for $\utraces$, the same is also possible as long as the two specifications are mutually accepting. 

This is also where the $\composable$ requirement becomes important. It enforces that all labels are either synchronised or only present in one of the two label sets. This means that every trace $\sigma$ can be split into a unique pair of two projected traces $\project{\sigma}{L_s^\delta}$ and $\project{\sigma}{L_e^\delta}$ which can be replayed in $s$ and $e$, respectively. Without this requirement, it would be unclear what to do with unsynchronised shared labels.

\begin{lemmarep}
let $i_s,i_e$ be $\composable$ $IOTS$, $i_s'\in Q_{i_s}$, $i_e' \in Q_{i_e}, \sigma \in {L_{i_s\parcomp i_e}^\delta}^*$.
\centermath{i_s \parcomp i_e \xRightarrow{\sigma} i_s' \parcomp i_e' \ \iff\ 
i_s \xRightarrow{\project{\sigma}{L_{i_s}^\delta}} i_s' \:\land\:
i_e \xRightarrow{\project{\sigma}{L_{i_e}^\delta}} i_e'}
\label{lemma:project_from_parcomp_IOTS}
\end{lemmarep}
\begin{proof}

~\begin{case_distinction}
    \item[($\Longrightarrow$):] Proof by induction on $\sigma$. 
    \begin{case_distinction}
        \item[Base case:] $\sigma = \epsilon$. Let $\sigma_\tau,\sigma_\tau', \sigma_\tau'' \in \tau^*$\\
        We will proof a stronger statement, the main ($\Longrightarrow$) goal for all states, instead of just for the initial state: \[\forall p_1 \in Q_{i_s}, q_1 \in Q_{i_e}: p_1\parcomp q_1 \xRightarrow{\epsilon} p \parcomp q \implies p_1 \xRightarrow{\project{\epsilon}{L_{i_s}^\delta}} p \land q_1 \xRightarrow{\project{\epsilon}{L_{i_e}^\delta}} q)\]
        $p_1 \parcomp q_1 \xRightarrow{\epsilon} p \parcomp q$\\
        \proofstep{\Cref{def:arrowdefs}: $\xRightarrow{\epsilon}$}\\
        $p_1 \parcomp q_1 \xrightarrow{\sigma_\tau} p \parcomp q$\\
        \proofstep{\Itemref{lemma:parcomp_base_properties}{step_tau}}\\
        $p_1 \xrightarrow{\sigma_\tau'} p \land q_1 \xrightarrow{\sigma_\tau''} q$\\
        \proofstep{\Cref{def:arrowdefs}: $\xRightarrow{\epsilon}$}\\
        $p_1 \xRightarrow{\epsilon} p \land q_1 \xRightarrow{\epsilon} q$\\
        \proofstep{\Cref{def:projection}: $\projectop$}\\
        $p_1 \xRightarrow{\project{\epsilon}{L_{i_s}^\delta}} p \land q_1 \xRightarrow{\project{\epsilon}{L_{i_e}^\delta}} q$\\
        
        \item[Induction step:]Assume the proposition holds for $\sigma'\in  {L_{i_s\parcomp i_e}^{\delta}}^{*}$.\\
        To proof: the proposition holds for $\sigma$, where $\sigma =\sigma' \cdot a$ with $a \in L_{i_s \parcomp i_e}^\delta$. This is divided into three cases based on $a$:
        \begin{itemize}[align=left]

            \item[$a \in L_{i_s} \setminus L_{i_e}$:]\ \\
            $i_s \parcomp i_e \xRightarrow{\sigma' \cdot a} i_s' \parcomp i_e'$\\
            \proofstep{\Cref{def:arrowdefs}: $\xRightarrow{\sigma}$}\\
            $\exists p_1, p_2 \in Q_{i_s}, q_1,q_2 \in Q_{i_e} :
            i_s \parcomp i_e \xRightarrow{\sigma'} p_1 \parcomp q_1 \land
            p_1 \parcomp q_1 \xrightarrow{a} p_2 \parcomp q_2 \land
            p_2 \parcomp q_2 \xRightarrow{\epsilon} i_s' \parcomp i_e' $\\
            \proofstep{Apply base case}\\
            $\exists p_1, p_2 \in Q_{i_s}, q_1,q_2 \in Q_{i_e} :
            i_s \parcomp i_e \xRightarrow{\sigma'} p_1 \parcomp q_1 \land
            p_1 \parcomp q_1 \xrightarrow{a} p_2 \parcomp q_2 \land
            p_2 \xRightarrow{\epsilon} s' \land
            q_2 \xRightarrow{\epsilon} e'$\\
            \proofstep{\cref{def:parcomp}: $q_1 = q_2$.
            From induction on structure of $T_{s\parcomp e}$ with $a \in L_{i_s} \setminus L_{i_e}$.}\\
        $\exists p_1, p_2 \in Q_{i_s}, q_1\in Q_{i_e} :
            i_s \parcomp i_e \xRightarrow{\sigma'} p_1 \parcomp q_1 \land
            p_1 \parcomp q_1 \xrightarrow{a} p_2 \parcomp q_1 \land
            p_2 \xRightarrow{\epsilon} s' \land
            q_1 \xRightarrow{\epsilon} e'$\\
            \proofstep{\Itemref{lemma:parcomp_base_properties}{step_left}}\\
            $\exists p_1, p_2 \in Q_{i_s}, q_1 \in Q_{i_e} :
            i_s \parcomp i_e \xRightarrow{\sigma'} p_1 \parcomp q_1 \land
            p_1  \xrightarrow{a} p_2 \land
            p_2 \xRightarrow{\epsilon} s' \land
            q_1 \xRightarrow{\epsilon} e'$\\
            \proofstep{\Cref{def:arrowdefs}: $\xRightarrow{a}$}\\
            $\exists p_1 \in Q_{i_s}, q_1 \in Q_{i_e} :
            i_s \parcomp i_e \xRightarrow{\sigma'} p_1 \parcomp q_1 \land
            p_1 \xRightarrow{a} s' \land
            q_1 \xRightarrow{\epsilon} e'$\\
            \proofstep{Apply IH}\\
            $\exists p_1 \in Q_{i_s}, q_1 \in Q_{i_e} :
            s \xRightarrow{\project{\sigma'}{L_{i_s}^\delta}} p_1 \land
            e \xRightarrow{\project{\sigma'}{L_{i_e}^\delta}} q_1 \land
            p_1 \xRightarrow{a} s' \land
            q_1 \xRightarrow{\epsilon} e'$\\
            \proofstep{\Itemref{prop:general_properties}{trans_transitive}}\\
            $s \xRightarrow{\project{\sigma'}{L_{i_s}^\delta}\cdot a} s' \land
            e \xRightarrow{\project{\sigma'}{L_{i_e}^\delta}} e'$\\
            \proofstep{\Cref{def:projection}: $\projectop$}\\
            $s \xRightarrow{\project{\sigma'\cdot a}{L_{i_s}^\delta}} s' \land
            e \xRightarrow{\project{\sigma'\cdot a}{L_{i_e}^\delta}} e'$\\
            
            \item[$a \in L_{i_e} \setminus L_{i_s}$:]Symmetric with the previous case.
            
            \item[$a\in (L_{i_s} \cap L_{i_e}) \cup \{\delta\}$:]\ \\
            $i_s \parcomp i_e \xRightarrow{\sigma' \cdot a} i_s' \parcomp i_e'$\\
            \proofstep{\Cref{def:arrowdefs}: $\xRightarrow{\sigma}$}\\
            $\exists p_1, p_2 \in Q_{i_s}, q_1,q_2 \in Q_{i_e} :
            i_s \parcomp i_e \xRightarrow{\sigma'} p_1 \parcomp q_1 \land
            p_1 \parcomp q_1 \xrightarrow{a} p_2 \parcomp q_2 \land
            p_2 \parcomp q_2 \xRightarrow{\epsilon} i_s' \parcomp i_e' $\\
            \proofstep{Apply base case}\\
            $\exists p_1, p_2 \in Q_{i_s}, q_1,q_2 \in Q_{i_e} :
            i_s \parcomp i_e \xRightarrow{\sigma'} p_1 \parcomp q_1 \land
            p_1 \parcomp q_1 \xrightarrow{a} p_2 \parcomp q_2 \land
            p_2 \xRightarrow{\epsilon} s' \land
            q_2 \xRightarrow{\epsilon} e' $\\
            \proofstep{\Cref{lemma:IOTS_step_both}}\\
            $\exists p_1, p_2 \in Q_{i_s}, q_1,q_2 \in Q_{i_e} :
            i_s \parcomp i_e \xRightarrow{\sigma'} p_1 \parcomp q_1 \land
            p_1 \xrightarrow{a} p_2 \land
            q_1 \xrightarrow{a} q_2 \land
            p_2 \xRightarrow{\epsilon} s' \land
            q_2 \xRightarrow{\epsilon} e' $\\
            \proofstep{\Cref{def:arrowdefs}: $\xRightarrow{a}$}\\
            $\exists p_1 \in Q_{i_s}, q_1 \in Q_{i_e} :
            i_s \parcomp i_e \xRightarrow{\sigma'} p_1 \parcomp q_1 \land
            p_1 \xRightarrow{a} s' \land
            q_1 \xRightarrow{a} e' $\\
            \proofstep{Apply IH}\\
            $\exists p_1 \in Q_{i_s}, q_1 \in Q_{i_e} :
            s \xRightarrow{\project{\sigma'}{L_{i_s}^\delta}} p_1 \land e \xRightarrow{\project{\sigma'}{L_{i_e}^\delta}} q_1 \land
            p_1 \xRightarrow{a} s' \land
            q_1 \xRightarrow{a} e' $\\
            \proofstep{\Itemref{prop:general_properties}{trans_transitive}}\\
            $s \xRightarrow{\project{\sigma'}{L_{i_s}^\delta}\cdot a} s' \land e \xRightarrow{\project{\sigma'}{L_{i_e}^\delta} \cdot a} e' $\\
            \proofstep{\Cref{def:projection}: $\projectop$}\\
            $s \xRightarrow{\project{\sigma'\cdot a}{L_{i_s}^\delta}} s' \land e \xRightarrow{\project{\sigma'\cdot a}{L_{i_e}^\delta}} e' $\\
          
        \end{itemize}
    \end{case_distinction}
    \item[($\Longleftarrow$):] Covered by \cref{lemma:project_from_parcomp_light}.

\end{case_distinction}
\end{proof}

\begin{lemmarep}
let $s,e$ be $\composable$ $LTS$, $s'\in Q_s$, $e' \in Q_e, \sigma \in  \utraces[s\parcomp e]$.
\centermath{s \mutuallyaccepts e \ \implies\ 
(\ s \parcomp e \xRightarrow{\sigma} s' \parcomp e' \ \iff\ 
s \xRightarrow{\project{\sigma}{L_s^\delta}} s' \:\land\:
e \xRightarrow{\project{\sigma}{L_e^\delta}} e'\ )}
\label{lemma:project_from_parcomp}
\end{lemmarep}

\begin{proof}
~\begin{case_distinction}
    \item[($\Longrightarrow$):] Proof by induction on $\sigma$. 
    \begin{case_distinction}
        \item[Base case:] $\sigma = \epsilon$. Let $\sigma_\tau,\sigma_\tau', \sigma_\tau'' \in \tau^*$\\
        We will proof a stronger statement, the main ($\Longrightarrow$) goal for all states, instead of just for the initial state: \[\forall p_1 \in Q_s, q_1 \in Q_e: p_1\parcomp q_1 \xRightarrow{\epsilon} p \parcomp q \implies p_1 \xRightarrow{\project{\epsilon}{L_s^\delta}} p \land q_1 \xRightarrow{\project{\epsilon}{L_e^\delta}} q)\]
        $p_1 \parcomp q_1 \xRightarrow{\epsilon} p \parcomp q$\\
        \proofstep{\Cref{def:arrowdefs}: $\xRightarrow{\epsilon}$}\\
        $p_1 \parcomp q_1 \xrightarrow{\sigma_\tau} p \parcomp q$\\
        \proofstep{\Itemref{lemma:parcomp_base_properties}{step_tau}}\\
        $p_1 \xrightarrow{\sigma_\tau'} p \land q_1 \xrightarrow{\sigma_\tau''} q$\\
        \proofstep{\Cref{def:arrowdefs}: $\xRightarrow{\epsilon}$}\\
        $p_1 \xRightarrow{\epsilon} p \land q_1 \xRightarrow{\epsilon} q$\\
        \proofstep{\Cref{def:projection}: $\projectop$}\\
        $p_1 \xRightarrow{\project{\epsilon}{L_s^\delta}} p \land q_1 \xRightarrow{\project{\epsilon}{L_s^\delta}} q$\\
        
        \item[Induction step:] Assume the proposition holds for $\sigma'\in  \utraces[s \parcomp e]$.\\
        To proof: the proposition holds for $\sigma$, where $\sigma =\sigma' \cdot a$ with $a \in L_{s \parcomp e}^\delta$. Assume $\sigma \in  \utraces[s \parcomp e]$, otherwise the proposition trivially holds. This is divided into three cases based on $a$:
        \begin{case_distinction}

            \item[$a \in L_s \setminus L_e$:]\ \\
            $s \parcomp e \xRightarrow{\sigma' \cdot a} s' \parcomp e'$\\
            \proofstep{\Cref{def:arrowdefs}: $\xRightarrow{\sigma}$}\\
            $\exists p_1, p_2 \in Q_s, q_1,q_2 \in Q_e :
            s \parcomp e \xRightarrow{\sigma'} p_1 \parcomp q_1 \land
            p_1 \parcomp q_1 \xrightarrow{a} p_2 \parcomp q_2 \land
            p_2 \parcomp q_2 \xRightarrow{\epsilon} s' \parcomp e' $\\
            \proofstep{Apply base case}\\
            $\exists p_1, p_2 \in Q_s, q_1,q_2 \in Q_e :
            s \parcomp e \xRightarrow{\sigma'} p_1 \parcomp q_1 \land
            p_1 \parcomp q_1 \xrightarrow{a} p_2 \parcomp q_2 \land
            p_2 \xRightarrow{\epsilon} s' \land
            q_2 \xRightarrow{\epsilon} e'$\\
            \proofstep{\cref{def:parcomp}: $q_1 = q_2$.
            From induction on structure of $T_{s\parcomp e}$ with $a \in L_s \setminus L_e$.}\\
        $\exists p_1, p_2 \in Q_s, q_1\in Q_e :
            s \parcomp e \xRightarrow{\sigma'} p_1 \parcomp q_1 \land
            p_1 \parcomp q_1 \xrightarrow{a} p_2 \parcomp q_1 \land
            p_2 \xRightarrow{\epsilon} s' \land
            q_1 \xRightarrow{\epsilon} e'$\\
            \proofstep{\Itemref{lemma:parcomp_base_properties}{step_left}}\\
            $\exists p_1, p_2 \in Q_s, q_1 \in Q_e :
            s \parcomp e \xRightarrow{\sigma'} p_1 \parcomp q_1 \land
            p_1  \xrightarrow{a} p_2 \land
            p_2 \xRightarrow{\epsilon} s' \land
            q_1 \xRightarrow{\epsilon} e'$\\
            \proofstep{\Cref{def:arrowdefs}: $\xRightarrow{a}$}\\
            $\exists p_1 \in Q_s, q_1 \in Q_e :
            s \parcomp e \xRightarrow{\sigma'} p_1 \parcomp q_1 \land
            p_1 \xRightarrow{a} s' \land
            q_1 \xRightarrow{\epsilon} e'$\\
            \proofstep{Apply IH}\\
            $\exists p_1 \in Q_s, q_1 \in Q_e :
            s \xRightarrow{\project{\sigma'}{L_s^\delta}} p_1 \land
            e \xRightarrow{\project{\sigma'}{L_e^\delta}} q_1 \land
            p_1 \xRightarrow{a} s' \land
            q_1 \xRightarrow{\epsilon} e'$\\
            \proofstep{\Itemref{prop:general_properties}{trans_transitive}}\\
            $s \xRightarrow{\project{\sigma'}{L_s^\delta}\cdot a} s' \land
            e \xRightarrow{\project{\sigma'}{L_e^\delta}} e'$\\
            \proofstep{\Cref{def:projection}: $\projectop$}\\
            $s \xRightarrow{\project{\sigma'\cdot a}{L_s^\delta}} s' \land
            e \xRightarrow{\project{\sigma'\cdot a}{L_e^\delta}} e'$\\
            
            \item[$a \in L_e \setminus L_s$:]Symmetric with the previous case.
            
            \item[$a\in (L_s \cap L_e) \cup \{\delta\}$:]\ \\
            $s \parcomp e \xRightarrow{\sigma' \cdot a} s' \parcomp e'$\\
            \proofstep{\Cref{def:arrowdefs}: $\xRightarrow{\sigma}$}\\
            $\exists p_1, p_2 \in Q_s, q_1,q_2 \in Q_e :
            s \parcomp e \xRightarrow{\sigma'} p_1 \parcomp q_1 \land
            p_1 \parcomp q_1 \xrightarrow{a} p_2 \parcomp q_2 \land
            p_2 \parcomp q_2 \xRightarrow{\epsilon} s' \parcomp e' $\\
            \proofstep{Apply base case}\\
            $\exists p_1, p_2 \in Q_s, q_1,q_2 \in Q_e :
            s \parcomp e \xRightarrow{\sigma'} p_1 \parcomp q_1 \land
            p_1 \parcomp q_1 \xrightarrow{a} p_2 \parcomp q_2 \land
            p_2 \xRightarrow{\epsilon} s' \land
            q_2 \xRightarrow{\epsilon} e' $\\
            \proofstep{\Itemref{lemma:parcomp_base_properties}{step_both}}\\
            $\exists p_1, p_2 \in Q_s, q_1,q_2 \in Q_e :
            s \parcomp e \xRightarrow{\sigma'} p_1 \parcomp q_1 \land
            p_1 \xrightarrow{a} p_2 \land
            q_1 \xrightarrow{a} q_2 \land
            p_2 \xRightarrow{\epsilon} s' \land
            q_2 \xRightarrow{\epsilon} e' $\\
            \proofstep{\Cref{def:arrowdefs}: $\xRightarrow{a}$}\\
            $\exists p_1 \in Q_s, q_1 \in Q_e :
            s \parcomp e \xRightarrow{\sigma'} p_1 \parcomp q_1 \land
            p_1 \xRightarrow{a} s' \land
            q_1 \xRightarrow{a} e' $\\
            \proofstep{Apply IH}\\
            $\exists p_1 \in Q_s, q_1 \in Q_e :
            s \xRightarrow{\project{\sigma'}{L_s^\delta}} p_1 \land e \xRightarrow{\project{\sigma'}{L_e^\delta}} q_1 \land
            p_1 \xRightarrow{a} s' \land
            q_1 \xRightarrow{a} e' $\\
            \proofstep{\Itemref{prop:general_properties}{trans_transitive}}\\
            $s \xRightarrow{\project{\sigma'}{L_s^\delta}\cdot a} s' \land e \xRightarrow{\project{\sigma'}{L_e^\delta} \cdot a} e' $\\
            \proofstep{\Cref{def:projection}: $\projectop$}\\
            $s \xRightarrow{\project{\sigma'\cdot a}{L_s^\delta}} s' \land e \xRightarrow{\project{\sigma'\cdot a}{L_e^\delta}} e' $\\
          
        \end{case_distinction}
    \end{case_distinction}
    \item[($\Longleftarrow$):] Covered by \cref{lemma:project_from_parcomp_light}.

\end{case_distinction}
\end{proof}

The $\mutuallyaccepts$ relation, and by extension \cref{lemma:project_from_parcomp}, only consider $\utraces$ and not arbitrary traces because only states reachable by $\utraces$ are important for $\uioco$.
This does, however, create the extra requirement of checking that after projecting a trace to a specific component, it is still part of the $\utraces$ for that component. \Cref{lemma:project_utraces} shows that the $\mutuallyaccepts$ relation ensures that $\utraces$ are preserved when projecting from a composed system, in both directions. This is not trivial, as the special label $\delta$ is not normally preserved under composition.

\begin{lemmarep} Let $s,e \in LTS$ be $\composable$, $\sigma \in {L_{s\parcomp e}^\delta}^*$.
\centermath{\begin{array}[t]{ll}
s \mutuallyaccepts e \ \implies \\
~(\ \sigma \in \utraces[s \parcomp e] \ \iff\ \project{\sigma}{L_s^\delta} \in  \utraces[s] \;\land\; \project{\sigma}{L_e^\delta}\in  \utraces[e]\ )
\end{array}}
\label{lemma:project_utraces}
\end{lemmarep}

\begin{proof}
\ \\
Assume $s \mutuallyaccepts e$. 
\begin{case_distinction}
    \item[($\Longrightarrow$):]\ \\
    If $ \utraces[s\parcomp e] = \emptyset$, then the lemma trivially holds.\\
    Assume $\sigma \in  \utraces[s\parcomp e]$, then the proof follows by induction on $\sigma$.\\
    \begin{case_distinction}
        \item[Base case:] $\sigma = \epsilon$\\
        $\epsilon \in  \utraces[s\parcomp e]$\\
        \proofstep{\Cref{def:uioco}: $Utraces$}\\
        $\exists p_1 \in Q_s, q_1 \in Q_e : s \parcomp e \xRightarrow{\epsilon} p_1 \parcomp q_1$\\
        \proofstep{\Cref{lemma:project_from_parcomp}}\\
        $\exists p_1 \in Q_s, q_1 \in Q_e : s\xRightarrow{\project{\epsilon}{L_s^\delta}} p_1 \land e\xRightarrow{\project{\epsilon}{L_e^\delta}} q_1$\\
        \proofstep{\Cref{def:uioco}: $Utraces$}\\
        $\project{\epsilon}{L_s^\delta} \in  \utraces[s] \land \project{\epsilon}{L_e^\delta} \in  \utraces[e]$\\
        
        \item[Induction step:] Assume the proposition holds for $\sigma'$. To proof: the proposition holds for $\sigma$  where $\sigma = \sigma' \cdot a$, $a \in L_{s \parcomp e}^\delta $.\\
        We do a case distinction on $a$, splitting $L_{s \parcomp e}^\delta$ into\\
        $(I_s \setminus L_e), (I_e \setminus L_s), (I_e \cap U_s), (I_s \cap U_e), (I_s \cap I_e), (U_s \setminus L_e), (U_e \setminus L_s), \{\delta\}$. Note that $U_s \cap U_e$ is missing because it is empty (\cref{def:composable}). This fine grained distinction is required because what it means for a trace to be part of $\utraces$ changes based on if the last label is an input or output. Additionally, the behaviour of $\projectop$ changes based on if a label is part of $L_s$, $L_e$ or both. 
        \begin{case_distinction}
        \item[$a \in I_s \setminus L_e$:]\ \\
            $\sigma'\cdot a \in  \utraces[s\parcomp e]$\\
            \proofstep{\Itemref{prop:general_properties}{prefix_closed}}\\
            $\sigma'\in  \utraces[s\parcomp e] \land \sigma'\cdot a \in  \utraces[s\parcomp e]$\\
            \proofstep{\Cref{def:uioco}: $Utraces$}\\
            $\sigma'\in  \utraces[s\parcomp e]\; \land $\\
            $\exists p \in Q_s, q \in Q_e : s\parcomp e \xRightarrow{\sigma'\cdot a} p \parcomp q\; \land$\\ 
            $\forall p_2 \in Q_s, q_2 \in Q_e : s \parcomp e \xRightarrow{\sigma'} p_2 \parcomp q_2 \implies p_2 \parcomp q_2 \xRightarrow{a}$\\
            \proofstep{\Itemref{prop:general_properties}{trans_transitive}}\\
            $\sigma'\in  \utraces[s\parcomp e] \land $\\
            $\exists p_1 \in Q_s, q_1 \in Q_e : s\parcomp e \xRightarrow{\sigma'} p_1 \parcomp q_1\; \land$\\
            $\forall p_2 \in Q_s, q_2 \in Q_e : s \parcomp e \xRightarrow{\sigma'} p_2 \parcomp q_2 \implies p_2 \parcomp q_2 \xRightarrow{a}$\\
            \proofstep{\Cref{lemma:project_from_parcomp}}\\
            $\sigma'\in  \utraces[s\parcomp e]\; \land $\\
            $\exists q_1 \in Q_e : e \xRightarrow{\project{\sigma'}{L_e^\delta}} q_1\; \land$\\
            $ \forall p_2 \in Q_s, q_2 \in Q_e : s \xRightarrow{\project{\sigma'}{L_s^\delta}} p_2 \land e \xRightarrow{\project{\sigma'}{L_e^\delta}}  q_2 \implies p_2 \parcomp q_2 \xRightarrow{a}$\\
            \proofstep{$\forall$ elimination}\\
            $\sigma'\in  \utraces[s\parcomp e]\; \land $\\
            $\exists q_1 \in Q_e : e \xRightarrow{\project{\sigma'}{L_e^\delta}} q_1\; \land$\\
            $ \forall p_2 \in Q_s : s \xRightarrow{\project{\sigma'}{L_s^\delta}} p_2 \land e \xRightarrow{\project{\sigma'}{L_e^\delta}}  q_1 \implies p_2 \parcomp q_1 \xRightarrow{a} $\\
            \proofstep{($A \land B \implies C) \land B \implies (A \implies C)$}\\
            $\sigma'\in  \utraces[s\parcomp e]\; \land $\\
            $\exists q_1 \in Q_e, \forall p_2 \in Q_s : s \xRightarrow{\project{\sigma'}{L_s^\delta}} p_2 \implies p_2 \parcomp q_1 \xRightarrow{a}$\\
            \proofstep{\cref{def:parcomp}: $\parcomp$}\\
            $\sigma'\in  \utraces[s\parcomp e]\; \land $\\
            $\forall p_2 \in Q_s : s \xRightarrow{\project{\sigma'}{L_s^\delta}} p_2 \implies p_2 \xRightarrow{a}$\\
            \proofstep{Apply IH}\\
            $\project{\sigma'}{L_s^\delta}\in  \utraces[s] \land \project{\sigma'}{L_e^\delta} \in  \utraces[e]\; \land$\\
            $\forall p_2 \in Q_s : s \xRightarrow{\project{\sigma'}{L_s^\delta}} p_2 \implies p_2 \xRightarrow{a}$\\
            \proofstep{\Cref{def:uioco}: $Utraces$}\\
            $\project{\sigma'}{L_s^\delta}\cdot a \in  \utraces[s] \land \project{\sigma'}{L_e^\delta} \in  \utraces[e]$\\
            \proofstep{\Cref{def:projection}: $\projectop$}\\
            $\project{\sigma'\cdot a}{L_s^\delta} \in  \utraces[s] \land \project{\sigma'\cdot a}{L_e^\delta} \in  \utraces[e]$
        \item[$a \in I_e \setminus L_s$:] Symmetric to previous case\\

        \item[$a \in I_s \cap U_e$:]\ \\
            $\sigma'\cdot a \in  \utraces[s\parcomp e]$\\
            \proofstep{\Itemref{prop:general_properties}{prefix_closed}}\\
            $\sigma'\in  \utraces[s\parcomp e] \land \sigma'\cdot a \in  \utraces[s\parcomp e]$\\
            \proofstep{\Cref{def:uioco}: $Utraces$}\\
            $\sigma'\in  \utraces[s\parcomp e]\; \land $\\
            $\exists p \in Q_s, q \in Q_e : s\parcomp e \xRightarrow{\sigma'\cdot a} p \parcomp q$\\
            \proofstep{\Cref{lemma:project_from_parcomp} }\\
            $\sigma'\in  \utraces[s\parcomp e]\; \land $\\
            $\exists q \in Q_e :  e \xRightarrow{\project{\sigma'\cdot a}{L_e^\delta}} q$\\
            \proofstep{\Cref{def:mutually_accepts,def:accepting}: $\mutuallyaccepts$ and $\accepts$}\\
            $\sigma'\in  \utraces[s\parcomp e]\; \land $\\
            $\exists q \in Q_e :  e \xRightarrow{\project{\sigma'\cdot a}{L_e^\delta}} q\; \land$\\
            $(\forall \sigma \in  \utraces[s\parcomp e], p_1 \in Q_s, q_1 \in Q_e: s \parcomp e \xRightarrow{\sigma} p_1 \parcomp q_1 \implies$\\
            $out(q_1)\cap I_e \subseteq in(p_1) \cap U_s$)\\
            \proofstep{$\forall$ elimination}\\
            $\sigma'\in  \utraces[s\parcomp e]\; \land $\\
            $\exists q \in Q_e :  e \xRightarrow{\project{\sigma'\cdot a}{L_e^\delta}} q\; \land$\\
            $\forall p_1, \in Q_s, q_1 \in Q_e : s \parcomp e \xRightarrow{\sigma'} p_1 \parcomp q_1 \implies out(q_1) \cap I_s \subseteq in(p_1) \cap U_e$\\
            \proofstep{\Cref{lemma:project_from_parcomp}}\\
            $\sigma'\in  \utraces[s\parcomp e]\; \land $\\
            $\exists q \in Q_e :  e \xRightarrow{\project{\sigma'\cdot a}{L_e^\delta}} q\; \land$\\
            $\forall p_1, \in Q_s, q_1 \in Q_e : s \xRightarrow{\project{\sigma'}{L_s^\delta}} p_1 \land e \xRightarrow{\project{\sigma'}{L_e^\delta}} q_1 \implies out(q_1) \cap I_s \subseteq in(p_1) \cap U_e$\\
            \proofstep{\Cref{def:out,def:inset}: $in$ and $out$, and $a \in I_s \cap U_e$}\\
            $\sigma'\in  \utraces[s\parcomp e]\; \land $\\
            $\exists q \in Q_e :  e \xRightarrow{\project{\sigma'\cdot a}{L_e^\delta}} q\; \land$\\
            $\forall p_1, \in Q_s, q_1 \in Q_e : s \xRightarrow{\project{\sigma'}{L_s^\delta}} p_1 \land e \xRightarrow{\project{\sigma'}{L_e^\delta}} q_1 \implies (q_1 \xrightarrow{a} \implies p_1 \xRightarrow{a})$\\
            \proofstep{\Cref{def:projection}: $\projectop$}\\
            $\sigma'\in  \utraces[s\parcomp e]\; \land $\\
            $\exists q \in Q_e :  e \xRightarrow{\project{\sigma'}{L_e^\delta}\cdot a} q\; \land$\\
            $\forall p_1, \in Q_s, q_1 \in Q_e : s \xRightarrow{\project{\sigma'}{L_s^\delta}} p_1 \land e \xRightarrow{\project{\sigma'}{L_e^\delta}} q_1 \implies (q_1 \xrightarrow{a} \implies p_1 \xRightarrow{a})$\\
            \proofstep{\cref{def:arrowdefs}: $e \xRightarrow{\project{\sigma'}{L_e^\delta}\cdot a} q$}\\
            $\sigma'\in  \utraces[s\parcomp e]\; \land $\\
            $\exists q, q_2 \in Q_e :  e \xRightarrow{\project{\sigma'}{L_e^\delta}\cdot a} q \land e \xRightarrow{\project{\sigma'}{L_e^\delta}} q_2 \land  q_2 \xrightarrow{a}\; \land$\\
            $\forall p_1, \in Q_s, q_1 \in Q_e : s \xRightarrow{\project{\sigma'}{L_s^\delta}} p_1 \land e \xRightarrow{\project{\sigma'}{L_e^\delta}} q_1 \implies (q_1 \xrightarrow{a} \implies p_1 \xRightarrow{a})$\\
            \proofstep{$\forall$ elimination}\\
            $\sigma'\in  \utraces[s\parcomp e]\; \land $\\
            $\exists q, q_2 \in Q_e : e \xRightarrow{\project{\sigma'}{L_e^\delta}\cdot a} q \land e \xRightarrow{\project{\sigma'}{L_e^\delta}} q_2 \land  q_2 \xrightarrow{a}\;\land$\\
            $\forall p_1, \in Q_s : s \xRightarrow{\project{\sigma'}{L_s^\delta}} p_1 \land e \xRightarrow{\project{\sigma'}{L_e^\delta}} q_2 \implies (q_2 \xrightarrow{a} \implies p_1 \xRightarrow{a})$\\
            \proofstep{($A \land B \implies C) \land B \implies (A \implies C)$}\\
            $\sigma'\in  \utraces[s\parcomp e]\; \land $\\
            $\exists q, q_2 \in Q_e : e \xRightarrow{\project{\sigma'}{L_e^\delta}\cdot a} q \land q_2 \xrightarrow{a}\; \land$\\
            $\forall p_1, \in Q_s: s \xRightarrow{\project{\sigma'}{L_s^\delta}} p_1 \implies (q_2 \xrightarrow{a} \implies p_1 \xRightarrow{a})$\\
            \proofstep{($A \implies B) \land A \implies B$}\\
            $\sigma'\in  \utraces[s\parcomp e]\; \land $\\
            $\exists q \in Q_e : e \xRightarrow{\project{\sigma'}{L_e^\delta}\cdot a} q\; \land$\\
            $\forall p_1, \in Q_s : s \xRightarrow{\project{\sigma'}{L_s^\delta}} p_1 \implies p_1 \xRightarrow{a}$\\
            \proofstep{Apply IH}\\
            $\project{\sigma'}{L_s^\delta}\in  \utraces[s] \land \project{\sigma'}{L_e^\delta} \in  \utraces[e] \land$\\
            $\exists q \in Q_e : e \xRightarrow{\project{\sigma'}{L_e^\delta}\cdot a} q\; \land$\\
            $\forall p_1, \in Q_s : s \xRightarrow{\project{\sigma'}{L_s^\delta}} p_1 \implies p_1 \xRightarrow{a}$\\
            \proofstep{\Cref{def:uioco}: $Utraces$}\\
            $\project{\sigma'}{L_s^\delta}\cdot a \in  \utraces[s] \land \project{\sigma'}{L_e^\delta}\cdot a \in  \utraces[e]$\\
            \proofstep{\Cref{def:projection}: $\projectop$}\\
            $\project{\sigma'\cdot a}{L_s^\delta} \in  \utraces[s] \land \project{\sigma'\cdot a}{L_e^\delta} \in  \utraces[e]$
            
        \item[$a \in I_e \cap U_s$:] Symmetric to previous case.\\

        \item[$a \in I_e \cap I_s$:]\ \\
            $\sigma'\cdot a \in  \utraces[s\parcomp e]$\\
            \proofstep{\Itemref{prop:general_properties}{prefix_closed}}\\
            $\sigma'\in  \utraces[s\parcomp e] \land \sigma'\cdot a \in  \utraces[s\parcomp e]$\\
            \proofstep{\Cref{def:uioco}: $Utraces$}\\
            $\sigma'\in  \utraces[s\parcomp e]\; \land $\\
            $\exists p \in Q_s, q \in Q_e : s\parcomp e \xRightarrow{\sigma'\cdot a} p \parcomp q\; \land$\\ 
            $\forall p_2 \in Q_s, q_2 \in Q_e : s \parcomp e \xRightarrow{\sigma'} p_2 \parcomp q_2 \implies p_2 \parcomp q_2 \xRightarrow{a}$\\
            \proofstep{\Itemref{prop:general_properties}{trans_transitive}}\\
            $\sigma'\in  \utraces[s\parcomp e] \land $\\
            $\exists p_1 \in Q_s, q_1 \in Q_e : s\parcomp e \xRightarrow{\sigma'} p_1 \parcomp q_1\; \land$\\
            $\forall p_2 \in Q_s, q_2 \in Q_e : s \parcomp e \xRightarrow{\sigma'} p_2 \parcomp q_2 \implies p_2 \parcomp q_2 \xRightarrow{a}$\\
            \proofstep{\Cref{lemma:project_from_parcomp}}\\
            $\sigma'\in  \utraces[s\parcomp e]\; \land $\\
            $\exists p_1 \in Q_s, q_1 \in Q_e : s \xRightarrow{\project{\sigma'}{L_s^\delta}} p_1 \land e \xRightarrow{\project{\sigma'}{L_e^\delta}} q_1\; \land$\\
            $ \forall p_2 \in Q_s, q_2 \in Q_e : s \xRightarrow{\project{\sigma'}{L_s^\delta}} p_2 \land e \xRightarrow{\project{\sigma'}{L_e^\delta}}  q_2 \implies p_2 \parcomp q_2 \xRightarrow{a}$\\
            \proofstep{$\forall$ elimination X2}\\
            $\sigma'\in  \utraces[s\parcomp e]\; \land $\\
            $\exists p_1 \in Q_s, q_1 \in Q_e  : s \xRightarrow{\project{\sigma'}{L_s^\delta}} p_1 \land e \xRightarrow{\project{\sigma'}{L_e^\delta}} q_1\; \land$\\
            $ \forall p_2 \in Q_s : s \xRightarrow{\project{\sigma'}{L_s^\delta}} p_2 \land e \xRightarrow{\project{\sigma'}{L_e^\delta}}  q_1 \implies p_2 \parcomp q_1 \xRightarrow{a}\;\land$\\
            $ \forall q_2 \in Q_e : s \xRightarrow{\project{\sigma'}{L_s^\delta}} p_1 \land e \xRightarrow{\project{\sigma'}{L_e^\delta}}  q_2 \implies p_1 \parcomp q_2 \xRightarrow{a}$\\
            \proofstep{($A \land B \implies C) \land B \implies (A \implies C)$}\\
            $\sigma'\in  \utraces[s\parcomp e]\; \land $\\
            $\exists q_1 \in Q_e, \forall p_2 \in Q_s : s \xRightarrow{\project{\sigma'}{L_s^\delta}} p_2 \implies p_2 \parcomp q_1 \xRightarrow{a}\;\land$\\
            $\exists p_1 \in Q_s, \forall q_2 \in Q_e : e \xRightarrow{\project{\sigma'}{L_e^\delta}}  q_2 \implies p_1 \parcomp q_2 \xRightarrow{a}$\\
            \proofstep{\cref{def:parcomp}: $\parcomp$}\\
            $\sigma'\in  \utraces[s\parcomp e]\; \land $\\
            $\forall p_2 \in Q_s : s \xRightarrow{\project{\sigma'}{L_s^\delta}} p_2 \implies p_2 \xRightarrow{a}\;\land$\\
            $\forall q_2 \in Q_e : e \xRightarrow{\project{\sigma'}{L_e^\delta}}  q_2 \implies q_2 \xRightarrow{a}$\\
            \proofstep{Apply IH}\\
            $\project{\sigma'}{L_s^\delta}\in  \utraces[s] \land \project{\sigma'}{L_e^\delta} \in  \utraces[e]\; \land$\\
            $\forall p_2 \in Q_s : s \xRightarrow{\project{\sigma'}{L_s^\delta}} p_2 \implies p_2 \xRightarrow{a}\;\land$\\
            $\forall q_2 \in Q_e : e \xRightarrow{\project{\sigma'}{L_e^\delta}}  q_2 \implies q_2 \xRightarrow{a}$\\\
            \proofstep{\Cref{def:uioco}: $Utraces$}\\
            $\project{\sigma'}{L_s^\delta}\cdot a \in  \utraces[s] \land \project{\sigma'}{L_e^\delta}\cdot a \in  \utraces[e]$\\
            \proofstep{\Cref{def:projection}: $\projectop$}\\
            $\project{\sigma'\cdot a}{L_s^\delta} \in  \utraces[s] \land \project{\sigma'\cdot a}{L_e^\delta} \in  \utraces[e]$
            
        \item[$a \in U_s \setminus L_e$:]\ \\
            $\sigma'\cdot a \in  \utraces[s\parcomp e]$\\
            \proofstep{\Itemref{prop:general_properties}{prefix_closed}}\\
            $\sigma'\in  \utraces[s\parcomp e] \land \sigma'\cdot a \in  \utraces[s\parcomp e]$\\
            \proofstep{\Cref{def:uioco}: $Utraces$}\\
            $\sigma'\in  \utraces[s\parcomp e] \land \sigma'\cdot a \in  \utraces[s\parcomp e]\; \land $\\
            $\exists p \in Q_s, q \in Q_e : s\parcomp e \xRightarrow{\sigma'\cdot a} p \parcomp q$\\
            \proofstep{\Cref{lemma:project_from_parcomp} }\\
            $\sigma'\in  \utraces[s\parcomp e]\; \land $\\
            $\exists p \in Q_s : s \xRightarrow{\project{\sigma'\cdot a}{L_s^\delta}} p$\\
            \proofstep{\Cref{def:projection}: $\projectop$}\\
            $\sigma'\in  \utraces[s\parcomp e]\; \land $\\
            $\exists p \in Q_s: s \xRightarrow{\project{\sigma'}{L_s^\delta}\cdot a} p$\\
            \proofstep{Apply IH}\\
            $\project{\sigma'}{L_s^\delta}\in  \utraces[s] \land \project{\sigma'}{L_e^\delta} \in  \utraces[e]\; \land$\\
            $\exists p \in Q_s : s \xRightarrow{\project{\sigma'}{L_s^\delta}\cdot a} p$\\
            \proofstep{\Cref{def:uioco}: $Utraces$}\\
            $\project{\sigma'}{L_s^\delta}\cdot a \in  \utraces[s] \land \project{\sigma'}{L_e^\delta}\in  \utraces[e]$\\
            \proofstep{\Cref{def:projection}: $\projectop$}\\
            $\project{\sigma'\cdot a}{L_s^\delta} \in  \utraces[s] \land \project{\sigma'\cdot a}{L_e^\delta} \in  \utraces[e]$
            
        \item[$a \in U_e \setminus L_s$:]Symmetric to previous case.\\
            
        \item [$a = \delta$:]\ \\
            $\sigma'\cdot \delta \in  \utraces[s\parcomp e]$\\
            \proofstep{\Itemref{prop:general_properties}{prefix_closed}}\\
            $\sigma'\in  \utraces[s\parcomp e] \land \sigma'\cdot \delta \in  \utraces[s\parcomp e]$\\
            \proofstep{\Cref{def:uioco}: $Utraces$}\\
            $\sigma'\in  \utraces[s\parcomp e] \land \sigma'\cdot \delta \in  \utraces[s\parcomp e] \;\land $\\
            $\exists p \in Q_s, q \in Q_e : s\parcomp e \xRightarrow{\sigma'\cdot \delta} p \parcomp q$\\
            \proofstep{\Cref{lemma:project_from_parcomp} }\\
            $\sigma'\in  \utraces[s\parcomp e]\; \land $\\
            $\exists p \in Q_s, q \in Q_e : s \xRightarrow{\project{\sigma'\cdot \delta}{L_s^\delta}} p \land e \xRightarrow{\project{\sigma'\cdot \delta}{L_e^\delta}} q$\\
            \proofstep{\Cref{def:projection}: $\projectop$}\\
            $\sigma'\in  \utraces[s\parcomp e]\; \land $\\
            $\exists p \in Q_s, q \in Q_e : s \xRightarrow{\project{\sigma'}{L_s^\delta}\cdot \delta} p \land e \xRightarrow{\project{\sigma'}{L_e^\delta}\cdot \delta} q$\\
            \proofstep{Apply IH}\\
            $\project{\sigma'}{L_s^\delta}\in  \utraces[s] \land \project{\sigma'}{L_e^\delta} \in  \utraces[e]\; \land$\\
            $\exists p \in Q_s, q \in Q_e : s \xRightarrow{\project{\sigma'}{L_s^\delta}\cdot \delta} p \land e \xRightarrow{\project{\sigma'}{L_e^\delta}\cdot \delta} q$\\
            \proofstep{\Cref{def:uioco}: $Utraces$}\\
            $\project{\sigma'}{L_s^\delta}\cdot \delta \in  \utraces[s] \land \project{\sigma'}{L_e^\delta}\cdot \delta\in  \utraces[e]$\\
            \proofstep{\Cref{def:projection}: $\projectop$}\\
            $\project{\sigma'\cdot \delta}{L_s^\delta} \in  \utraces[s] \land \project{\sigma'\cdot \delta}{L_e^\delta} \in  \utraces[e]$
        \end{case_distinction}
    \end{case_distinction}
    \item[($\Longleftarrow$):]\ \\
     Assume $\project{\sigma}{L_s^\delta} \in  \utraces[s]\; \land\; \project{\sigma}{L_e^\delta} \in  \utraces[e]$, then the proof follows by induction on $\sigma$.\\
    \begin{case_distinction}
        \item[Base case:] $\sigma = \epsilon$\\
        
        $\project{\epsilon}{L_s^\delta} \in  \utraces[s] \land \project{\epsilon}{L_e^\delta} \in  \utraces[e]$\\
        \proofstep{\Cref{def:uioco}: $Utraces$}\\
        $\exists p_1 \in Q_s, q_1 \in Q_e : s\xRightarrow{\project{\epsilon}{L_s^\delta}} p_1 \land e\xRightarrow{\project{\epsilon}{L_e^\delta}} q_1$\\
        \proofstep{\Cref{lemma:project_from_parcomp_light}}\\
        $\exists p_1 \in Q_s, q_1 \in Q_e : s \parcomp e \xRightarrow{\epsilon} p_1 \parcomp q_1$\\
        \proofstep{\Cref{def:uioco}: $Utraces$}\\
        $\epsilon \in  \utraces[s\parcomp e]$\\
        
        \item[Induction step:]\ \\
        IH: the proposition holds for $\sigma'$. To proof: the proposition holds for $\sigma$ where $\sigma = \sigma' \cdot a$, $a \in L_{s \parcomp e}^\delta $.
        We do a case distinction on $a$, splitting $L_{s \parcomp e}^\delta$ into\\
        $(I_s \setminus L_e), (I_e \setminus L_s), ((I_s \cap U_e)\cup (I_e \cap U_s)\cup\{\delta\}), (I_s \cap I_e), (U_s \setminus L_e), (U_e \setminus L_s)$. Note that $U_s \cap U_e$ is missing because it is empty (\cref{def:composable}).:
        \begin{case_distinction}
            \item[$a \in I_s \setminus L_e$:]\ \\
            $\project{\sigma'\cdot a}{L_s^\delta} \in  \utraces[s] \land \project{\sigma'\cdot a}{L_e^\delta} \in  \utraces[e]$\\
            \proofstep{\Cref{def:projection}: $\projectop$}\\
            $\project{\sigma'}{L_s^\delta}\cdot a \in  \utraces[s] \land \project{\sigma'}{L_e^\delta} \in  \utraces[e]$\\
            \proofstep{\Cref{def:uioco}: $Utraces$}\\
            $\project{\sigma'}{L_s^\delta}\cdot a\in  \utraces[s] \land \project{\sigma'}{L_e^\delta} \in  \utraces[e]\; \land$\\
            $\exists p_1 \in Q_s, q_1 \in Q_e: s \xRightarrow{\project{\sigma'}{L_s^\delta}\cdot a} p_1\land e \xRightarrow{\project{\sigma'}{L_e^\delta}} q_1$\\
            $\forall p_2 \in Q_s : s \xRightarrow{\project{\sigma'}{L_s^\delta}} p_2 \implies p_2 \xRightarrow{a}$\\
            \proofstep{\Itemref{prop:general_properties}{prefix_closed}}\\
            $\project{\sigma'}{L_s^\delta}\in  \utraces[s] \land \project{\sigma'}{L_e^\delta} \in  \utraces[e]\; \land$\\
            $\exists p_1 \in Q_s, q_1 \in Q_e: s \xRightarrow{\project{\sigma'}{L_s^\delta}\cdot a} p_1\land e \xRightarrow{\project{\sigma'}{L_e^\delta}} q_1$\\
            $\forall p_2 \in Q_s : s \xRightarrow{\project{\sigma'}{L_s^\delta}} p_2 \implies p_2 \xRightarrow{a}$\\
            \proofstep{Apply IH}\\
            $\sigma'\in  \utraces[s\parcomp e]\; \land $\\
            $\exists p_1 \in Q_s, q_1 \in Q_e: s \xRightarrow{\project{\sigma'}{L_s^\delta}\cdot a} p_1\land e \xRightarrow{\project{\sigma'}{L_e^\delta}} q_1$\\
            $\forall p_2 \in Q_s : s \xRightarrow{\project{\sigma'}{L_s^\delta}} p_2 \implies p_2 \xRightarrow{a}$\\
            \proofstep{\Cref{def:projection}: $\projectop$}\\
            $\sigma'\in  \utraces[s\parcomp e]\; \land $\\
            $\exists p_1 \in Q_s, q_1 \in Q_e : s\xRightarrow{\project{\sigma'\cdot a}{L_s^\delta}} p_1 \land e \xRightarrow{\project{\sigma'\cdot a}{L_e^\delta}} q_1\; \land$\\
            $\forall p_2 \in Q_s: s \xRightarrow{\project{\sigma'}{L_s^\delta}} p_2 \implies p_2  \xRightarrow{\project{a}{L_s^\delta}}$\\
            \proofstep{$(A \implies C) \implies (A \land B \implies C)$}\\
            $\sigma'\in  \utraces[s\parcomp e]\; \land $\\
            $\exists p_1 \in Q_s, q_1 \in Q_e : s\xRightarrow{\project{\sigma'\cdot a}{L_s^\delta}} p_1 \land e \xRightarrow{\project{\sigma'\cdot a}{L_e^\delta}} q_1\; \land$\\
            $\forall p_2 \in Q_s, q_2 \in Q_e : s \xRightarrow{\project{\sigma'}{L_s^\delta}} p_2 \land  e \xRightarrow{\project{\sigma'}{L_e^\delta}} q_2 \implies p_2  \xRightarrow{\project{a}{L_s^\delta}}$\\
            \proofstep{\Cref{def:arrowdefs}: $\forall p: p \xRightarrow{\epsilon}$}\\
            $\sigma'\in  \utraces[s\parcomp e]\; \land $\\
            $\exists p_1 \in Q_s, q_1 \in Q_e : s\xRightarrow{\project{\sigma'\cdot a}{L_s^\delta}} p_1 \land e \xRightarrow{\project{\sigma'\cdot a}{L_e^\delta}} q_1\; \land$\\
            $\forall p_2 \in Q_s, q_2 \in Q_e : s \xRightarrow{\project{\sigma'}{L_s^\delta}} p_2 \land  e \xRightarrow{\project{\sigma'}{L_e^\delta}} q_2 \implies p_2  \xRightarrow{\project{a}{L_s^\delta}} \land q_2 \xRightarrow{\epsilon}$\\
            \proofstep{\Cref{def:projection}: $\projectop$}\\
            $\sigma'\in  \utraces[s\parcomp e]\; \land $\\
            $\exists p_1 \in Q_s, q_1 \in Q_e : s\xRightarrow{\project{\sigma'\cdot a}{L_s^\delta}} p_1 \land e \xRightarrow{\project{\sigma'\cdot a}{L_e^\delta}} q_1\; \land$\\
            $\forall p_2 \in Q_s, q_2 \in Q_e : s \xRightarrow{\project{\sigma'}{L_s^\delta}} p_2 \land  \xRightarrow{\project{\sigma'}{L_e^\delta}} q_2 \implies p_2  \xRightarrow{\project{a}{L_s^\delta}} \land q_2 \xRightarrow{\project{ a}{L_e^\delta}}$\\
            \proofstep{\Cref{lemma:project_from_parcomp_light}}\\
            $\sigma'\in  \utraces[s\parcomp e]\; \land $\\
            $\exists p_1 \in Q_s, q_1 \in Q_e : s\parcomp e \xRightarrow{\sigma'\cdot a} p_1 \parcomp q_1\; \land$\\ 
            $\forall p_2 \in Q_s, q_2 \in Q_e : s \parcomp e \xRightarrow{\sigma'} p_2 \parcomp q_2 \implies p_2 \parcomp q_2 \xRightarrow{a}$\\
            \proofstep{\Cref{def:uioco}: $Utraces$}\\
            $\sigma'\cdot a \in  \utraces[s\parcomp e]$\\

            \item[$a \in I_e \setminus L_s$:] Symmetric to previous case\\
            
            \item[$a \in (I_s \cap U_e)\cup (I_e \cap U_s)\cup\{\delta\}$:]\ \\
            $\project{\sigma'\cdot a}{L_s^\delta} \in  \utraces[s] \land \project{\sigma'\cdot a}{L_e^\delta} \in  \utraces[e]$\\
            \proofstep{\Cref{def:uioco}: $Utraces$}\\
            $\project{\sigma'\cdot a}{L_s^\delta}\in  \utraces[s] \land \project{\sigma'\cdot a}{L_e^\delta}\in  \utraces[e] \land$\\
            $\exists p_1 \in Q_s, q_1 \in Q_e: s \xRightarrow{\project{\sigma'\cdot a}{L_s^\delta}} p_1\land e \xRightarrow{\project{\sigma'\cdot a}{L_e^\delta}} q_1$\\
            \proofstep{\Cref{def:projection}: $\projectop$}\\
            $\project{\sigma'}{L_s^\delta}\cdot a\in  \utraces[s] \land \project{\sigma'}{L_e^\delta}\cdot a\in  \utraces[e] \land$\\
            $\exists p_1 \in Q_s, q_1 \in Q_e: s \xRightarrow{\project{\sigma'\cdot a}{L_s^\delta}} p_1\land e \xRightarrow{\project{\sigma'\cdot a}{L_e^\delta}} q_1$\\
            \proofstep{\Itemref{prop:general_properties}{prefix_closed}}\\
            $\project{\sigma'}{L_s^\delta}\in  \utraces[s] \land \project{\sigma'}{L_e^\delta} \in  \utraces[e] \land$\\
            $\exists p_1 \in Q_s, q_1 \in Q_e: s \xRightarrow{\project{\sigma'\cdot a}{L_s^\delta}} p_1\land e \xRightarrow{\project{\sigma'\cdot a}{L_e^\delta}} q_1$\\
            \proofstep{Apply IH}\\
            $\sigma'\in  \utraces[s\parcomp e] \land $\\
            $\exists p_1 \in Q_s, q_1 \in Q_e: s \xRightarrow{\project{\sigma'\cdot a}{L_s^\delta}} p_1\land e \xRightarrow{\project{\sigma'\cdot a}{L_e^\delta}} q_1$\\
            \proofstep{\Cref{lemma:project_from_parcomp_light}}\\
            $\sigma'\in  \utraces[s\parcomp e] \land $\\
            $\exists p_1 \in Q_s, q_1 \in Q_e : s\parcomp e \xRightarrow{\sigma'\cdot a} p_1 \parcomp q_1$\\ 
            \proofstep{\Cref{def:uioco}: $Utraces + a \notin I_{s\parcomp e}$}\\
            $\sigma'\cdot a \in  \utraces[s\parcomp e]$\\

            \item[$a \in I_s \cap I_e$:]\ \\
            $\project{\sigma'\cdot a}{L_s^\delta} \in  \utraces[s] \land \project{\sigma'\cdot a}{L_e^\delta} \in  \utraces[e]$\\
            \proofstep{\Cref{def:projection}: $\projectop$}\\
            $\project{\sigma'}{L_s^\delta}\cdot a \in  \utraces[s] \land \project{\sigma'}{L_e^\delta}\cdot a \in  \utraces[e]$\\
            \proofstep{\Cref{def:uioco}: $Utraces$}\\
            $\project{\sigma'}{L_s^\delta}\cdot a\in  \utraces[s] \land \project{\sigma'}{L_e^\delta}\cdot a \in  \utraces[e]\; \land$\\
            $\exists p_1 \in Q_s, q_1 \in Q_e: s \xRightarrow{\project{\sigma'}{L_s^\delta}\cdot a} p_1\land e \xRightarrow{\project{\sigma'}{L_e^\delta}\cdot a} q_1\;\land$\\
            $\forall p_2 \in Q_s : s \xRightarrow{\project{\sigma'}{L_s^\delta}} p_2 \implies p_2 \xRightarrow{a} \;\land$\\
            $\forall q_2 \in Q_e : e \xRightarrow{\project{\sigma'}{L_s^\delta}} q_2 \implies q_2 \xRightarrow{a}$\\
            \proofstep{\Itemref{prop:general_properties}{prefix_closed}}\\
            $\project{\sigma'}{L_s^\delta}\in  \utraces[s] \land \project{\sigma'}{L_e^\delta} \in  \utraces[e]\; \land$\\
            $\exists p_1 \in Q_s, q_1 \in Q_e: s \xRightarrow{\project{\sigma'}{L_s^\delta}\cdot a} p_1\land e \xRightarrow{\project{\sigma'}{L_e^\delta}\cdot a} q_1\;\land$\\
            $\forall p_2 \in Q_s : s \xRightarrow{\project{\sigma'}{L_s^\delta}} p_2 \implies p_2 \xRightarrow{a} \;\land$\\
            $\forall q_2 \in Q_e : e \xRightarrow{\project{\sigma'}{L_s^\delta}} q_2 \implies q_2 \xRightarrow{a}$\\
            \proofstep{Apply IH}\\
            $\sigma'\in  \utraces[s\parcomp e]\; \land $\\
            $\exists p_1 \in Q_s, q_1 \in Q_e: s \xRightarrow{\project{\sigma'}{L_s^\delta}\cdot a} p_1\land e \xRightarrow{\project{\sigma'}{L_e^\delta}\cdot a} q_1\;\land$\\
            $\forall p_2 \in Q_s : s \xRightarrow{\project{\sigma'}{L_s^\delta}} p_2 \implies p_2 \xRightarrow{a} \;\land$\\
            $\forall q_2 \in Q_e : e \xRightarrow{\project{\sigma'}{L_s^\delta}} q_2 \implies q_2 \xRightarrow{a}$\\
            \proofstep{$(A \implies C) \land (B \implies D) \implies (A \land B \implies C \land D)$}\\
            $\sigma'\in  \utraces[s\parcomp e]\; \land $\\
             $\exists p_1 \in Q_s, q_1 \in Q_e: s \xRightarrow{\project{\sigma'}{L_s^\delta}\cdot a} p_1\land e \xRightarrow{\project{\sigma'}{L_e^\delta}\cdot a} q_1\;\land$\\
            $\forall p_2 \in Q_s, q_2 \in Q_e : s \xRightarrow{\project{\sigma'}{L_s^\delta}} p_2 \land  e \xRightarrow{\project{\sigma'}{L_e^\delta}} q_2 \implies p_2  \xRightarrow{a}\land\; q_2 \xRightarrow{a}$\\
            \proofstep{\cref{def:projection}: $\projectop$}\\
            $\sigma'\in  \utraces[s\parcomp e]\; \land $\\
             $\exists p_1 \in Q_s, q_1 \in Q_e: s \xRightarrow{\project{\sigma'\cdot a}{L_s^\delta}} p_1\land e \xRightarrow{\project{\sigma'\cdot a}{L_e^\delta}} q_1\;\land$\\
            $\forall p_2 \in Q_s, q_2 \in Q_e : s \xRightarrow{\project{\sigma'}{L_s^\delta}} p_2 \land  e \xRightarrow{\project{\sigma'}{L_e^\delta}} q_2 \implies p_2  \xRightarrow{\project{a}{L_s^\delta}}\land q_2 \xRightarrow{\project{a}{L_e^\delta}}$\\
            \proofstep{\Cref{lemma:project_from_parcomp_light}}\\
            $\sigma'\in  \utraces[s\parcomp e]\; \land $\\
            $\exists p_1 \in Q_s, q_1 \in Q_e : s\parcomp e \xRightarrow{\sigma'\cdot a} p_1 \parcomp q_1\; \land$\\ 
            $\forall p_2 \in Q_s, q_2 \in Q_e : s \parcomp e \xRightarrow{\sigma'} p_2 \parcomp q_2 \implies p_2  \parcomp q_2 \xRightarrow{a}$\\
            \proofstep{\Cref{def:uioco}: $Utraces$}\\
            $\sigma'\cdot a \in  \utraces[s\parcomp e]$\\
            
            \item[$a \in (U_s \setminus L_e)$:]\ \\
            $\project{\sigma'\cdot a}{L_s^\delta} \in  \utraces[s] \land \project{\sigma'\cdot a}{L_e^\delta} \in  \utraces[e]$\\
            \proofstep{\Cref{def:uioco}: $Utraces$}\\
            $\project{\sigma'\cdot a}{L_s^\delta} \in  \utraces[s] \land \project{\sigma'\cdot a}{L_e^\delta} \in  \utraces[e]\;\land$\\
            $\exists p \in Q_s, q \in Q_e : s \xRightarrow{\project{\sigma'\cdot a}{L_s^\delta}} p \land e \xRightarrow{\project{\sigma'\cdot a}{L_e^\delta}} q$\\
            \proofstep{\Cref{def:projection}: $\projectop$}\\
            $\project{\sigma'}{L_s^\delta}\cdot a \in  \utraces[s] \land \project{\sigma'}{L_e^\delta}\in  \utraces[e]\;\land$\\
            $\exists p \in Q_s, q \in Q_e : s \xRightarrow{\project{\sigma'\cdot a}{L_s^\delta}} p \land e \xRightarrow{\project{\sigma'\cdot a}{L_e^\delta}} q$\\
            \proofstep{\Itemref{prop:general_properties}{prefix_closed}}\\
            $\project{\sigma'}{L_s^\delta} \in  \utraces[s] \land \project{\sigma'}{L_e^\delta} \in  \utraces[e]\; \land$\\
            $\exists p \in Q_s, q \in Q_e : s \xRightarrow{\project{\sigma'\cdot a}{L_s^\delta}} p \land e \xRightarrow{\project{\sigma'\cdot a}{L_e^\delta}} q$\\
            \proofstep{Apply IH}\\
            $\sigma'\in  \utraces[s\parcomp e]\; \land $\\
            $\exists p \in Q_s, q \in Q_e : s \xRightarrow{\project{\sigma'\cdot a}{L_s^\delta}} p \land e \xRightarrow{\project{\sigma'\cdot a}{L_e^\delta}} q$\\
            \proofstep{\Cref{lemma:project_from_parcomp_light} }\\
            $\sigma'\in  \utraces[s\parcomp e]\; \land $\\
            $\exists p \in Q_s, q \in Q_e : s\parcomp e \xRightarrow{\sigma'\cdot a} p \parcomp q$\\
            \proofstep{\Cref{def:uioco}: $Utraces$}\\
            $\sigma'\cdot a \in  \utraces[s\parcomp e]$\\
            
            \item[$a \in (U_e \setminus L_s)$:]Symmetric to previous case.\\
        \end{case_distinction}
    \end{case_distinction}
\end{case_distinction}
\end{proof}

We now present \Cref{theo:eco_comp_accepting}, which is the main statement of this paper. It states that for mutually accepting specifications, $\uioco$ is preserved under parallel composition. The other way around, if there is a problem that causes two composed implementations to be $\uioco$-incorrect to their composed specifications, then this problem can also be found by testing with at least one of the components. The reverse of this implication, however, does not hold, even if both specifications are mutually accepting.

The reason for this is that the mutual acceptance relation only guarantees that no invalid outputs are communicated. It does not enforce that something is actually communicated. Therefore, it is possible for one of the two implementations to produce quiescence when this is not allowed, which is then masked in the combined system by the outputs generated by the other component. This highlights a property of the $\uioco$ relation: presence of specific outputs cannot be enforced. One possible way to deal with this might be to extend the $\uioco$ theory with a more fine grained concept of quiescence, allowing the detection of quiescence in specific components, instead of only over the whole system. This is further explored in \cite{noroozi_ImprovingInputOutputConformance_2014}.

\begin{theoremrep}
Let $s,e \in LTS$ be $\composable$, $i_s, i_e\in IOTS$, then\\
    \[s \mutuallyaccepts e \ \land\  
    i_s \uioco s \ \land\  i_e \uioco e \ \implies
    \ i_s \parcomp i_e \uioco s \parcomp e\]
    
\label{theo:eco_comp_accepting}
\end{theoremrep}

\begin{proof}
\ \\
Assume $e \mutuallyaccepts s$. Take $\sigma \in  \utraces[s\parcomp e], x \in out(i_s \parcomp i_e \after \sigma)$. If either $\sigma$ or $x$ does not exist ($ \utraces[s\parcomp e] = \emptyset$ or $out(i_s \parcomp i_e \after \sigma) = \emptyset$), then the theorem trivially holds. To prove: $x \in out(s\parcomp e \after \sigma)$\\\\
$i_s \uioco s \land i_e \uioco e \land$\\
$\sigma \in  \utraces[s\parcomp e]\land x \in out(i_s \parcomp i_e \after \sigma)$\\
\proofstep{\Cref{def:out,def:after}: $out$ and $\after$}\\
$i_s \uioco s \land i_e \uioco e \land \sigma \in  \utraces[s\parcomp e]\;\land$\\
$\exists p_1 \in Q_{i_s}, q_1 \in Q_{i_e} : i_s \parcomp i_e \xRightarrow{\sigma\cdot x} p_1 \parcomp q_1$\\
\proofstep{\Cref{lemma:project_from_parcomp_IOTS}}\\
$i_s \uioco s \land i_e \uioco e \land \sigma \in  \utraces[s\parcomp e]\;\land$\\
$\exists p_1 \in Q_{i_s}, q_1 \in Q_{i_e} : i_s \xRightarrow{\project{\sigma\cdot x}{L_{i_s}^\delta}} p_1  \land  i_e \xRightarrow{\project{\sigma\cdot x}{L_{i_e}^\delta}}  q_1$\\
\proofstep{Case destinction on $x$}
\begin{case_distinction}
    \item[$x = \delta$:]\ \\
    \proofstep{\Cref{def:projection}: $\projectop$}\\
    $i_s \uioco s \land i_e \uioco e \land \sigma \in  \utraces[s\parcomp e]\;\land$\\
    $\exists p_1 \in Q_{i_s}, q_1 \in Q_{i_e} : i_s \xRightarrow{\project{\sigma}{L_{i_s}^\delta}\cdot \delta} p_1  \land  i_e \xRightarrow{\project{\sigma}{L_{i_e}^\delta}\cdot \delta}  q_1$\\
    \proofstep{\Cref{def:out,def:after}: $out$ and $\after$ }\\
    $i_s \uioco s \land i_e \uioco e \land \sigma \in  \utraces[s\parcomp e]\;\land$\\
    $\delta \in out(i_s \after \project{\sigma}{L_{i_s}^\delta}) \land \delta \in out(i_e \after \project{\sigma}{L_{i_e}^\delta})$\\
    \proofstep{\Cref{lemma:project_utraces}}\\
    $i_s \uioco s \land i_e \uioco e \land \sigma \in  \utraces[s\parcomp e]\land \project{\sigma}{L_s^\delta} \in  \utraces[s]\land \project{\sigma}{L_e^\delta} \in  \utraces[e]\; \land $\\
    $\delta \in out(i_s \after \project{\sigma}{L_{i_s}^\delta}) \land \delta \in out(i_e \after \project{\sigma}{L_{i_e}^\delta})$\\
    \proofstep{\Cref{def:uioco}: $\uioco$ and $L_{i_s} = L_s$}\\
    $\sigma \in  \utraces[s\parcomp e]\; \land $\\
    $\delta \in out(s \after \project{\sigma}{L_s^\delta}) \land \delta \in out(e \after \project{\sigma}{L_e^\delta})$\\
    \proofstep{\Cref{def:out,def:after}: $out$ and $\after$ }\\
    $\sigma \in  \utraces[s\parcomp e] \land $\\
    $\exists p_1 \in Q_s, q_1 \in Q_e : s \xRightarrow{\project{\sigma}{L_s^\delta}\cdot \delta} p_1  \land  e \xRightarrow{\project{\sigma}{L_e^\delta}\cdot \delta}  q_1$\\
    \proofstep{\Cref{def:arrowdefs}: $\xRightarrow{\sigma \cdot \delta}$ }\\
    $\sigma \in  \utraces[s\parcomp e]\; \land $\\
    $\exists p_1, p_2 \in Q_s, q_1, q_2 \in Q_e : s \xRightarrow{\project{\sigma}{L_s^\delta}} p_2 \land p_2 \xrightarrow{\delta} p_1  \land  e \xRightarrow{\project{\sigma}{L_e^\delta}}  q_2 \land q_2 \xrightarrow{\delta} q_1$\\
    \proofstep{\Cref{lemma:project_from_parcomp}}\\
    $\sigma \in  \utraces[s\parcomp e]\; \land $\\
    $\exists p_1, p_2 \in Q_s, q_1, q_2 \in Q_e : s\parcomp e \xRightarrow{\sigma} p_2 \parcomp q_2 \land p_2 \xrightarrow{\delta} p_1 \land q_2 \xrightarrow{\delta} q_1$\\
    \proofstep{\Cref{def:parcomp}: $\parcomp$}\\
    $\sigma \in  \utraces[s\parcomp e]\; \land $\\
    $\exists p_1, p_2 \in Q_s, q_1, q_2 \in Q_e : s\parcomp e \xRightarrow{\sigma} p_2 \parcomp q_2 \land p_2 \parcomp q_2 \xrightarrow{\delta} p_1 \parcomp q_1$\\
    \proofstep{\Cref{def:out,def:after}: $out$ and $\after$ }\\
    $\delta \in out(s\parcomp e \after \sigma)$

    \item[$x \neq \delta$:]\ \\
    \proofstep{\cref{def:projection}: $\projectop$ + $x \in U_{i_s} \lor x \in U_{i_e}$}\\
    $i_s \uioco s \land i_e \uioco e \land \sigma \in  \utraces[s\parcomp e]\;\land$\\
    $\exists p_1 \in Q_{i_s}, q_1 \in Q_{i_e} : (i_s \xRightarrow{\project{\sigma}{L_{i_s}^\delta}\cdot x} p_1 \land x \in U_{I_s})  \lor (i_e \xRightarrow{\project{\sigma}{L_{i_e}^\delta}\cdot x}  q_1\land x \in U_{I_e})$\\
    \proofstep{\Cref{def:out,def:after}: $out$ and $\after$}\\
    $i_s \uioco s \land i_e \uioco e \land \sigma \in  \utraces[s\parcomp e]\;\land$\\
    $(x \in out(i_s\after \project{\sigma}{L_{i_s}^\delta}) \lor x \in out(i_e \after \project{\sigma}{L_{i_e}^\delta}))$\\
    \proofstep{\Cref{lemma:project_utraces}}\\
    $i_s \uioco s \land i_e \uioco e \land \sigma \in  \utraces[s\parcomp e]\land \project{\sigma}{L_s^\delta} \in  \utraces[s]\land \project{\sigma}{L_e^\delta} \in  \utraces[e]\;\land$\\
    $(x \in out(i_s\after \project{\sigma}{L_{i_s}^\delta}) \lor x \in out(i_e \after \project{\sigma}{L_{i_e}^\delta}))$\\
    \proofstep{\Cref{def:uioco}: $\uioco$ and $L_{i_s} = L_s$}\\
    $\sigma \in  \utraces[s\parcomp e]\;\land$\\
    $(x \in out(s\after \project{\sigma}{L_s^\delta}) \lor x \in out(e \after \project{\sigma}{L_e^\delta}))$\\
    \proofstep{\Cref{def:uioco}: $\utraces$ }\\
    $\sigma \in  \utraces[s\parcomp e]\;\land$\\
    $\exists p_1 \in Q_s, q_1 \in Q_e: s \parcomp e \xRightarrow{\sigma} p_1 \parcomp q_1$\\
    $(x \in out(s\after \project{\sigma}{L_s^\delta}) \lor x \in out(e \after \project{\sigma}{L_e^\delta}))$\\
    \proofstep{\Cref{lemma:project_from_parcomp}: }\\
    $\sigma \in  \utraces[s\parcomp e]\;\land$\\
    $\exists p_1 \in Q_s, q_1 \in Q_e: s \xRightarrow{\project{\sigma}{L_s^\delta}} p_1  \land s \xRightarrow{\project{\sigma}{L_e^\delta}} q_1 $\\
    $(x \in out(s\after \project{\sigma}{L_s^\delta}) \lor x \in out(e \after \project{\sigma}{L_e^\delta}))$\\
    \proofstep{\Cref{def:out,def:after}: $out$ and $\after$ }\\
    $\sigma \in  \utraces[s\parcomp e]\;\land$\\
    $\exists p_1 \in Q_s, q_1 \in Q_e: s \xRightarrow{\project{\sigma}{L_s^\delta}} p_1  \land e \xRightarrow{\project{\sigma}{L_e^\delta}} q_1 $\\
    $\exists p_2 \in Q_s, q_2 \in Q_e: ( (s \xRightarrow{ \project{\sigma}{L_s^\delta}} p_2 \land p_2 \xrightarrow{x}) \lor (e \xRightarrow{ \project{\sigma}{L_e^\delta}} q_2 \land q_2 \xrightarrow{x}))$\\
    \proofstep{\Cref{lemma:project_from_parcomp}: }\\
    $\sigma \in  \utraces[s\parcomp e]\;\land$\\
    $\exists p_1, p_2 \in Q_s, q_1, q_2 \in Q_e: ( (s \parcomp e \xRightarrow{ \sigma} p_2 \parcomp q_1 \land p_2 \xrightarrow{x}) \lor (s \parcomp e \xRightarrow{ \sigma} p_1 \parcomp q_2 \land q_2 \xrightarrow{x}))$\\
    \proofstep{Further case distinction on $x$: $x \in U_{s\parcomp e} \implies x \in (U_s \setminus L_e) \cup (U_e \setminus L_s) \cup (L_s \cap L_e)$}
    \begin{case_distinction}
        \item[$x \in U_s \setminus L_e$:]\ \\
        $\exists p_1, p_2 \in Q_s, q_1, q_2 \in Q_e: (s \parcomp e \xRightarrow{ \sigma} p_2 \parcomp q_1 \land p_2 \xrightarrow{x})$\\
        \proofstep{\Itemref{lemma:parcomp_base_properties}{step_left}}\\
        $\exists p_1, p_2 \in Q_s, q_1, q_2 \in Q_e: (s \parcomp e \xRightarrow{ \sigma} p_2 \parcomp q_1 \land p_2 \parcomp q_1 \xrightarrow{x})$\\
        \proofstep{\Cref{def:out,def:after}: $out$ and $after$}\\
        $x \in out(s\parcomp e \after \sigma)$\\
        
        \item[$x \in U_e \setminus L_s$:] Symmetric to previous case.
        \item[$x \in L_s \cup L_e$:]\ \\
        \proofstep{\Cref{def:mutually_accepts,def:accepting}: $\accepts$ and $\mutuallyaccepts$}\\
        $\exists p_1, p_2 \in Q_s, q_1, q_2 \in Q_e: out(p_2) \cap I_e \subseteq in(q_1) \cap U_s \land out(q_2) \cap i_s \subseteq in(p_1) \cap U_s\;\land$\\
        $( (s \parcomp e \xRightarrow{ \sigma} p_2 \parcomp q_1 \land p_2 \xrightarrow{x}) \lor (s \parcomp e \xRightarrow{ \sigma} p_1 \parcomp q_2 \land q_2 \xrightarrow{x}))$\\
        \proofstep{\Cref{def:out,def:inset}: $out$ and $in$}\\
        $\exists p_1, p_2 \in Q_s, q_1, q_2 \in Q_e:$\\
        \tab$(s \parcomp e \xRightarrow{ \sigma} p_2 \parcomp q_1 \land p_2 \xrightarrow{x} \land\; q_1 \xRightarrow{x})\; \lor$\\
        \tab$ (s \parcomp e \xRightarrow{ \sigma} p_1 \parcomp q_2 \land q_2 \xrightarrow{x} \land\; p_1 \xRightarrow{x})$\\
        \proofstep{\Cref{def:arrowdefs}: $\xRightarrow{x}$}\\
        $\exists p_1, p_2, p_3 \in Q_s, q_1, q_2, q_3 \in Q_e:$\\
        \tab$(s \parcomp e \xRightarrow{ \sigma} p_2 \parcomp q_1 \land p_2 \xrightarrow{x} \land\; q_1 \xRightarrow{\epsilon} q_3 \land q_3 \xrightarrow{x})\; \lor$\\
        \tab$ (s \parcomp e \xRightarrow{ \sigma} p_1 \parcomp q_2 \land q_2 \xrightarrow{x} \land\; p_1 \xRightarrow{\epsilon} p_3 \land p_3 \xrightarrow{x})$\\
        \proofstep{\Itemref{lemma:parcomp_base_properties}{step_tau}}\\
        $\exists p_1, p_2, p_3 \in Q_s, q_1, q_2, q_3 \in Q_e:$\\
        \tab$ (s \parcomp e \xRightarrow{ \sigma} p_2 \parcomp q_1 \land p_2 \xrightarrow{x} \land\; p_2 \parcomp q_1 \xRightarrow{\epsilon} p_2 \parcomp q_3 \land q_3 \xrightarrow{x})\;\lor$\\
        \tab$(s \parcomp e \xRightarrow{ \sigma} p_1 \parcomp q_2 \land q_2 \xrightarrow{x} \land\; p_1 \parcomp q_2 \xRightarrow{\epsilon} p_3 \parcomp q_2\land p_3 \xrightarrow{x})$\\
        \proofstep{\Itemref{prop:general_properties}{trans_transitive}}\\
        $\exists p_1, p_2, p_3 \in Q_s, q_1, q_2, q_3 \in Q_e:$\\
        \tab$(s \parcomp e \xRightarrow{ \sigma} p_2 \parcomp q_3 \land p_2 \xrightarrow{x} \land\; q_3 \xrightarrow{x})\; \lor$\\
        \tab$(s \parcomp e \xRightarrow{ \sigma} p_3 \parcomp q_2 \land q_2 \xrightarrow{x} \land\; p_3 \xrightarrow{x})$\\
        \proofstep{\Itemref{lemma:parcomp_base_properties}{both_to_parcomp}}\\
        $\exists p_1, p_2, p_3 \in Q_s, q_1, q_2, q_3 \in Q_e:$\\
        \tab$(s \parcomp e \xRightarrow{ \sigma} p_2 \parcomp q_3 \land p_2 \parcomp q_3 \xrightarrow{x})\;\lor$\\
        \tab$(s \parcomp e \xRightarrow{ \sigma} p_3 \parcomp q_2 \land p_3 \parcomp q_2 \xrightarrow{x})$\\
        \proofstep{\Cref{def:out,def:after}: $out$ and $after$}\\
        $x \in out(s\parcomp e \after \sigma)$\\
    \end{case_distinction}
\end{case_distinction}
\end{proof}

Another thing to note is that when applying \cref{theo:eco_comp_accepting} in practice, this makes the implicit assumption that you can correctly compose components. In order to guarantee the correctness of the composed system, the composition of components $i_s$ and $i_e$ must actually behave as $i_s\parcomp i_e$. This means that any communicating channels must be connected as described in $s$ and $e$, and that there must not be some hidden implicit environment part of the composition setup that further influences the behaviour of either of the components.

\equalizeCounters{}

\section{The Parking System Revisited}
\label{subsec:acceptingExample}

\equalizeCounters
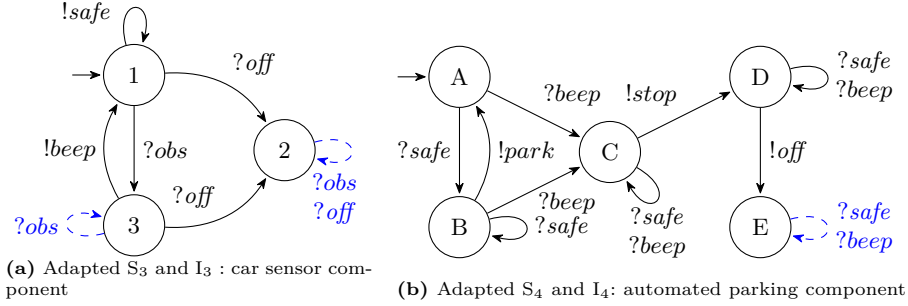
\begin{figure}[t!]
\begin{subfigure}[b]{.4\linewidth}
\phantomcaption\label{subfig:carsensor_Adapted}
\speclabel{sensor_Adapted}
\implabel{sensor_Adapted}
\begin{tikzpicture}[LTS]
\node[state,initial] (1) {1};
\node[state, below right= .5\nodedistance and \nodedistance of 1] (2) {2};
\node[state, below=of 1] (3) {3};

\path[->] 
    (1) edge [bend left] node [auto] {$\mathit{?off}$} (2)
        edge [] node [auto] {$?obs$} (3)
        edge [loop above] node [below left=-2mm and 2mm] {$\mathit{!safe}$} (1)
    (2) edge [loop right, imp only] node [below left=2mm and -2mm] {$?obs$\\$\mathit{?off}$} (2)
    (3) edge [bend right] node [auto] {$\mathit{?off}$} (2)
        edge [bend left] node [auto] {$!beep$} (1)
        edge [loop left, imp only] node [auto] {$?obs$} (3);

\end{tikzpicture}
\captiontext{Adapted \currentspec{} and \currentimp{} : car sensor component}

\end{subfigure}
\begin{subfigure}[b]{.59\linewidth}
\phantomcaption\label{subfig:carauto_park_Adapted}
\speclabel{auto_park_Adapted}
\implabel{auto_park_Adapted}
\begin{tikzpicture}[LTS]
\node[initial,state] (A) {A};
\node[state, below=of A] (B) {B};
\node[state, below right=.5\nodedistance and \nodedistance of A] (C) {C};
\node[state, above right=.5\nodedistance and \nodedistance of C] (D) {D};
\node[state, below=of D] (E) {E};

\path[->] 
    (A) edge node [left] {$\mathit{?safe}$} (B)
        edge node [auto] {$?beep$} (C)
    (B) edge [bend right] node [right] {$!park$} (A)
        edge [loop right] node [auto] {$\mathit{?safe}$} (B)
        edge node [below right= -1mm] {$?beep$} (C)
    (C) edge node [auto] {$!stop$} (D)
        edge [loop below right] node [below] {$\mathit{?safe}$\\$?beep$} (C)
    (D) edge node [auto] {$\mathit{!off}$} (E)
        edge [loop right] node [auto] {$\mathit{?safe}$\\$?beep$} (D)
    (E) edge [loop right, imp only] node [auto] {$\mathit{?safe}$\\$?beep$} (E);

\end{tikzpicture}
\captiontext{Adapted \currentspec{} and \currentimp{}: automated parking component}

\end{subfigure}
\caption{Mutually accepting versions of \cref{fig:carParkSensorComponents}. (spec: $\tikzarrow$, imp: $\tikzarrow[imp only]$)}
\label{fig:carParkSensorComponentsAdapted}
\end{figure}

Using the results from \cref{sec:definitions} and \cref{sec:acceptingSystems}, the problems with the parking system from \cref{sec:exampleIntro} can be explained: the two specifications in \cref{fig:carParkSensorComponents} are not mutually accepting. A counterexample is the trace $\mathit{safe}\cdot\mathit{obs}$, which is in the $\utraces$ of \cref{spec:composed_park}, and goes to state $B3$. In state 3, however, \cref{spec:sensor} can perform output $\mathit{beep}$, while \cref{spec:auto_park} does not accept input $\mathit{beep}$ in state $B$. A number of other counterexamples can also be given, each one corresponding to one of the dashed output transitions of \cref{spec:composed_park}. These are the states where the composed implementation produces unspecified outputs. Using these counterexamples, the points where the specifications have to be extended can be identified. The result can be seen in \cref{fig:carParkSensorComponentsAdapted}. \Cref{subfig:carauto_park_Adapted} now has several extra transitions for $\mathit{safe}$ and $\mathit{beep}$ defined in the specification, exactly in those places where the sensor might supply these inputs. The developer is now forced to think about what actually should happen there, while the developer is still free to not specify inputs that should not occur in normal operation. This is especially relevant for the $\mathit{beep}$ transition originating from state $B$, which was previously unspecified. On further inspection, it is revealed that a simple self loop is not desired here, because after a $\mathit{beep}$ the car should stop, and not continue to park. This would have resulted in undesired behaviour if the specifications were simply made input enabled in an automatic way, as was done in previous approaches \cite{vanderbijl_CompositionalTestingIoco_2004}. Using a self-loop here would mean that implementations are possible which pass all tests, but still do not stop when an object is detected.
\composedlabelPreloadFormat{spec:sensor_Adapted}{spec:auto_park_Adapted}
\composedlabelPreloadFormat{imp:sensor_Adapted}{imp:auto_park_Adapted}

The result of composing the adapted specifications from \cref{fig:carParkSensorComponentsAdapted} is shown in \cref{fig:carParkSensorComposedAdapted}. This specification now correctly finds that \cref{imp:composed_park} $\notuioco$ \cref{spec:composed_park_Adapted}, which can be seen with the trace $\mathit{safe}\cdot\mathit{obs}\cdot\mathit{beep}$ which is present in the $\utraces$ of \cref{spec:composed_park_Adapted}. After this trace, \cref{imp:composed_park} can produce the output $\mathit{park}$, which is undesirable after detecting an object, and also not allowed by \cref{spec:composed_park_Adapted}. But if each individual implementation is updated to be $\uioco$ correct according to its own adapted specification, as is done with \cref{imp:sensor_Adapted} and \cref{imp:auto_park_Adapted}, then their composition is again correct with respect to the composed specifications, i.e.\ \cref{imp:composed_park_Adapted} $\uioco$ \cref{spec:composed_park_Adapted}. 

\begin{figure}[t!]
\centering
\begin{tikzpicture}[LTS]
\node[initial,state] (A1) {A1};
\node[state, right of =  A1] (B1) {B1};
\node[state, below of= A1] (A3) {A3};
\node[state, right of= A3] (B3) {B3};
\node[state, right of= B3] (C1) {C1};
\node[state, above of= C1] (C3) {C3};
\node[state, right of= C1] (D1) {D1};
\node[state, right of= C3] (D3) {D3};
\node[state, below right = 0.5\nodedistance and \nodedistance of D3] (E2) {E2};

\path[->] 
    (A1)    edge[bend left] node [above] {$\mathit{!safe}$} (B1)
            edge node [left] {$?obs$} (A3)
    (B1)    edge[bend left] node [below] {$!park$} (A1)
            edge node [right] {$?obs$} (B3)
            edge[loop above] node [above] {$\mathit{!safe}$} (B1) 
    (B3)    edge node [above] {$!park$} (A3)
            edge node [above] {$!beep$} (C1)
            edge [loop below, imp only] node [below] {$?obs$} (B3)
    (A3)    edge [bend right=50, looseness=1.3] node [below] {$!beep$} (C1)
            edge [loop below,imp only] node[below] {$?obs$} (A3)
    (C1)    edge [bend left] node [above left] {$?obs$} (C3)
            edge node [below] {$!stop$} (D1)
            edge[loop below] node [below] {$\mathit{!safe}$} (C1)
    (C3)    edge[bend left] node [below right] {$!beep$} (C1)
            edge node [above] {$!stop$} (D3)
            edge [loop above, imp only] node [above] {$?obs$} (C3)
    (D1)    edge [bend left] node [above left] {$?obs$} (D3)
            edge node [below right] {$\mathit{!off}$} (E2)
            edge[loop below] node [below] {$\mathit{!safe}$} (D1)
    (D3)    edge [bend left] node [right] {$!beep$} (D1)
            edge node [above right] {$\mathit{!off}$} (E2)
            edge [loop above, imp only] node [above] {$?obs$} (D3)
    (E2)    edge [loop right, imp only] node [right] {$?obs$} (E2)
    ;

\end{tikzpicture}
\composedlabelNoFormat{spec:sensor_Adapted}{spec:auto_park_Adapted}{spec:composed_park_Adapted}
\composedlabelNoFormat{imp:sensor_Adapted}{imp:auto_park_Adapted}{imp:composed_park_Adapted}
\caption{Adapted car autopark and sensor composed (\cref{spec:sensor_Adapted}$\parcomp$\cref{spec:auto_park_Adapted} $\tikzarrow$) and (\cref{imp:sensor_Adapted}$\parcomp$\cref{imp:auto_park_Adapted} $\tikzarrow[imp only]$)}
\label{fig:carParkSensorComposedAdapted}
\end{figure}

The example shows how the $\mutuallyaccepts$ relation can be used to find integration problems between components using their specifications. Possible problems are prevented by expanding the specification, without requiring a full specification of all inputs. This does not yet require any actual implementations, as the reasoning is done over the domain of all possible valid implementations. Finding these integration problems before starting integration testing, allows for fixing them earlier in development. 

\section{Component Substitution and Diagnosis}
\label{sec:componentSubstitution}

In addition to providing compositionality in development and testing, mutual acceptance has benefits when retesting systems with updated components, and when diagnosing systems consisting of components.
A common situation is that one component becomes deprecated and needs to be replaced. This traditionally has a high cost, because even if the new component is well tested, there is a chance using it will cause problems with the other components already in use. These issues mainly occur because replacing a component changes the environment for the other components. This means the other components, which are the environment of the replaced component, might be called with new inputs which have not yet been tested. A well known example where reuse of an old, well tested component in a new environment caused the whole system to fail is the crash of the Ariane-5 rocket \cite{lions_ArianeFlight501_1996,weyuker_TestingComponentbasedSoftware_1998}. Here, an important subsystem put implicit requirements on the environment which were not documented or checked to hold. Correctness was inferred from extensive testing, but after changing the environment this testing became invalid, and the component failed anyway.

These problems can be reduced by using a specification-based analysis like $\mutuallyaccepts$ in combination with model-based testing. Model-based testing can generate tests for every defined sequence of inputs. If two specifications are mutually accepting, then they only communicate outputs which are defined inputs for the intended communication partner. These two points together mean that all the model-based testing done up to the point of replacing a component is still useful, because it was testing for all possible inputs, and not just the ones that were in current use. This can give a much higher confidence that a component switch will not cause any problems, because testing does not have to start from square one. If the specification of the new component is not mutually accepting with all the rest of the system, then the counterexamples point to all the places where undefined inputs are given. This information can be used to improve the specifications, and focus testing toward these possible problem areas.

The correctness reasoning made possible by the $\mutuallyaccepts$ relation can also be used during diagnosis, by taking the converse of \cref{theo:eco_comp_accepting}.  If the whole system contains a problem, and one or more components are found to be $\uioco$ correct, then the problem must be located within one of the remaining components. Together with \cref{lemma:project_utraces,lemma:project_from_parcomp} this can then be used to narrow down a trace showing $\uioco$ incorrectness of the whole system to a shorter trace showing $\uioco$ incorrectness of one specific component. This idea is expressed in \cref{lemma:project_diagnosis}. Since $\composable$ requires each label to be part of at most one output set, the last output of the counterexample uniquely identifies the problem component. This does not work if the last output was $\delta$, which could have been caused by a number of components. In this case we can still find the faulty component by replaying the projected traces in all components until the faulty one is found. This can, for example, more accurately determine the source of bugs from gathered logs containing full system traces.

\begin{lemmarep}
 Let $s,e \in LTS$, $i_s,i_e \in \IOTS$, $s\mutuallyaccepts e$, $\sigma \in \utraces[s \parcomp e]$.
\[\begin{array}[t]{lll}
    \sigma & \textit{is a counterexample for } i_s\parcomp i_e \uioco s \parcomp e & \implies \\
    & \project{\sigma}{L_s^\delta} \textit{ is a counterexample for } i_s \uioco s  & \ \lor \\
    & \project{\sigma}{L_e^\delta} \textit{ is a counterexample for } i_e \uioco e & \\
\end{array}\]
\label{lemma:project_diagnosis}
\end{lemmarep}
\begin{proof}
    \ \\
    We have $\sigma$ is a counterexample for $i_s\parcomp i_e \uioco s \parcomp e$. This means we have:
    $\exists \ell \in U_{s\parcomp e}^\delta, p_i \in Q_{s_i}, q_i \in Q_{e_i}, i_s\parcomp i_e \Trans{\sigma} p_i \parcomp q_i \land p_i\parcomp q_i \trans{\ell} \land$\\
    $\exists p \in Q_s, q \in Q_e, s\parcomp e \Trans{\sigma} p \parcomp q \land p \parcomp q \nottrans{\ell}$.\\
    By \cref{lemma:project_utraces} we have $\project{\sigma}{L_s^\delta} \in \utraces[s] \land \project{\sigma}{L_e^\delta} \in \utraces[e]$.\\
    Then with \cref{lemma:project_from_parcomp,lemma:project_from_parcomp_IOTS}, we have that after projection we still reach the same states. This gives:\\
    $\exists \ell \in U_{s\parcomp e}^\delta, p_i \in Q_{s_i}, q_i \in Q_{e_i}, i_s \Trans{\project{\sigma}{L_s^\delta}} p_i \land  i_e \Trans{\project{\sigma}{L_e^\delta}} q_i\land p_i\parcomp q_i \trans{\ell} \land$\\
    $\exists p \in Q_s, q \in Q_e, s \Trans{\project{\sigma}{L_s^\delta}} p  \land e \Trans{\project{\sigma}{L_e^\delta}} q   \land p \parcomp q \nottrans{\ell}$.\\
    Finally, we do structural induction on the definition of $p_i\parcomp q_i \trans{\ell}$ and $p \parcomp q \nottrans{\ell}$ (\cref{def:parcomp}). This leads to several cases. If $\ell$ is an unsynchronised output from $U_s$, this gives $p_i\trans{\ell}$ and $p\nottrans{\ell}$, which is a counterexample for $i_s\uioco s$. The same is true for $i_e \uioco e$ if we assume $\ell$ is an unsynchronised output from $U_e$. If we assume $\ell$ is synchronised instead, we follow the same reasoning, but now with the extra step that due to $s\mutuallyaccepts e$, $p\parcomp q \nottrans{\ell}$ still implies $\ell \notin \outset{p} \land \ell \notin \outset{q}$. The final case is quiescence. In this case we know either $p$ or $q$ or both are not quiescent: $p\nottrans{\ell} \lor\; q \nottrans{\ell}$. For $p_i$ and $q_i$ there are two options according to \cref{def:parcomp}: they are both quiescent or there is a communication problem. If they are both quiescent than this gives a counterexample for $\uioco$ correctness for at least one of them. If there is a communication problem then either $p_i$ or $q_i$ produces an output, and the other one does not accept this output as an input. This is not actually possible however, because both $s_i$ and $e_i$ are input complete, and therefore all their states accept all inputs.
\end{proof}

\section{Related Work}
\label{sec:related_work}
The work in this paper is closely related to ideas already discussed in the context of interface automata \cite{dealfaro_InterfaceAutomata_2001,dealfaro_InterfaceBasedDesign_2005,dealfaro_InterfaceTheoriesComponentBased_2001}. Interface automata are a type of labelled transition system which can be used to model both the behaviour of a component and the constraints it puts on its environment. These constraints are encoded in the form of missing input transitions, which then signify that the component can only be used in an environment that does not give these inputs. This closely resembles the main idea behind the $\mutuallyaccepts$ relation. Apart from a slightly different composability requirement which makes it associative, our definition for parallel composition coincides with the one from \cite{dealfaro_InterfaceAutomata_2001}. Our definition for $\mutuallyaccepts$ also seems to coincide with the absence of reachable (by $\utraces$) error states as defined  in \cite{dealfaro_InterfaceAutomata_2001}. The solution to reachable error states taken for interface automata is to apply a pruning algorithm. This will remove input transitions to further restrict the valid environments until all error states become unreachable. A downside of this approach is that it becomes easy to generate composed models that after pruning no longer give errors, but also no longer express the desired correct behaviour. This is noted in \cite{dealfaro_InterfaceAutomata_2001} as the observation that the environment that does not give any inputs at all, always avoids all avoidable error states. The interface automata approach consists of removing transitions from the composed specification until problem areas are unreachable. We instead choose to add transitions to the component specifications until the problem areas no longer exist. Another contribution of our work is the inclusion of quiescence, and the direct link to the $\uioco$ implementation relation. This makes the theory easier to apply in practice in the context of existing MBT tools.

An earlier attempt at formalising the correctness of a component with respect to its environment was developed in \cite{frantzen_ModelBasedTestingEnvironmental_2007}. It defines the $\eco$ (environmental conformance) relation with similar semantics to the accepts relation. The relation $\eco$, however, works on a specification for the environment, and a black box implementation of the component. This means that $\eco$ conformance can only be checked by testing, and this needs to be redone completely whenever a component changes. Additionally, all labels of the component and its environment have to communicate, i.e., there is no external communication,  which further restricts applicability. The $\eco$ approach also has a couple of advantages. Since $\eco$ is checked using testing, it can be done on the fly. It also does not require how a component calls other components as part of its input specifications. Instead, this information is gathered while testing and compared against the specifications of the components being called. This makes a possible combination of our work with the $\eco$ theory and algorithms interesting.

In this paper, we describe when a component is a valid environment for another component. Earlier work looking into the set of valid environments for a given component was done in the field of contract-based design. A detailed overview of this field can be found in \cite{benveniste_ContractsSystemDesign_2018}. A contract is defined as a tuple of a set of valid environments and a set of valid implementations, where every combination of environment and implementation can be composed. The definition of what it means for an environment to be composable with an implementation is very similar to our definitions, and it also describes how a labelled transition system can be seen as a contract. The scope under consideration in \cite{benveniste_ContractsSystemDesign_2018}, however, is limited to receptive environments with the same label set as the components. All components also have to be deterministic, and internal transitions or quiescence are not discussed. A more recent addition to contract theory extends the scope of \cite{benveniste_ContractsSystemDesign_2018} to hyper-contracts \cite{incer_InterfaceAutomataHypercontracts_2022}. While this extends the scope of properties that can be expressed as contracts, the current instantiation of the meta-theory for labelled transition systems still has many of the restrictions imposed in \cite{benveniste_ContractsSystemDesign_2018}. In contrast to the bottom up approach of combining component contracts into a composed contract, a top down approach is also possible and sometimes desired. Decomposing a set of requirements into individual component contracts has been studied in \cite{kaiser_ContractBasedDesignEmbedded_2015}.

Another way of describing compatible components is defining a specification for the most permissive communication partner. All concrete communication partners are then in some form of a refinement relation with this "operating guideline". This approach is outlined in \cite{massuthe_OperatingGuidelinesAutomatatheoretic_2005} for acyclic finite labelled transition systems. It assumes all communications to be asynchronous, while we assume synchronous communication.

\section{Future Work}
\label{sec:future_work}

Making specifications mutually accepting involves defining extra behaviour. Some of this extra specification is desirable, for instance the $\textit{beep}$ transition from state $B$ in \cref{spec:auto_park_Adapted}. This transition represents interesting behaviour that was missed in the specification phase. Most other added transitions, however, are just simple self-loops, which represent that the input has to be ignored. If receiving an input that was not specified is considered undefined behaviour, this is required to ensure correct behaviour. Another possible interpretation would be that unspecified inputs are buffered, until the other component is ready to receive them. In such a setting, it would not be required that every input that can be given is specified immediately. It would then be enough that such inputs are specified always eventually, after some amount of internal actions of the receiving component. In general, it can be investigated how to (automatically) repair non-mutually accepting systems.

In this paper, we have defined mutual acceptance, but no ways for practically checking it have been given. Algorithms to efficiently check mutual acceptance between specifications, or testing procedures to test mutual acceptance, analogous to $\eco$, need to be developed.

The theory introduced so far works on two components. Larger systems consist of many components. Mutual acceptance can still be inferred by repeatedly applying the parallel composition operator and \cref{theo:eco_comp_accepting}. For example, when combining specifications $s1,s2$ and $s3$ into $(s1\parcomp s2)\parcomp s3$, we must check that $s1\parcomp s2 \mutuallyaccepts s3$. Doing this directly using $s1\parcomp s2$ might be complicated due to the increasing number of parallel components. We postulate that multiway-mutual acceptance can be inferred from pairwise-mutual acceptance. In general, mutual acceptance of many components, with complicated communication structures, should be further investigated.

Requiring $\mutuallyaccepts$ for all intermediate steps means that there cannot be any unexpected outputs. For real systems however, these outputs are only a problem if they appear in the  final composition of all the components. The fact that two components do not work well in all environments is not a problem if you plan to use them together with other components that will prevent this. Therefore, a different definition of mutual acceptance for more than two components at a time might be investigated. 

To apply the theory in this paper to a practical use case, it will need to be extended with the concept of data. Real systems can seldom be modelled with a finite set of labels, but will instead send instances of data types to each other. This has been formalised in the theory of symbolic transition systems ($STS$) \cite{frantzen_TestGenerationBased_2005,vandenbos_CoverageBasedTestingSymbolic_2019}, which is the underlying formalism of several MBT tools. The concepts in this paper could be extended to $STS$ which would bring them closer to being applied in practice.

\section{Conclusion}
\label{sec:conclusion}

Model-based testing is a promising technology for increasing the efficiency and effectiveness of testing. The applicability of MBT, however, is limited by the availability of models. Larger system models are hard to create, but can be composed from multiple smaller component models. 
In this paper, we have defined the mutual acceptance relation $\mutuallyaccepts$ between specifications, which guarantees that model-based testing is compositional,
i.e.\ if two components have been tested for $\uioco$-correctness with respect to their respective specifications, then the composition of these implementations is also $\uioco$-correct with respect to the composition of their specifications, under the assumption that the parallel composition operator itself is faithfully implemented. This is an improvement over previous results which obtained the same conclusion with a stricter requirement, viz.\ that all specifications must be input-enabled \cite{vanderbijl_CompositionalTestingIoco_2004}.
In addition, we have shown that this result can also help when updating older components with newer ones, and when localising a faulty component during diagnosis of a large, component-based system. 

\newpage
\printbibliography

\end{document}